\documentclass[lettersize,journal]{IEEEtran}
\usepackage{amsmath,amsfonts,amssymb}
\usepackage{algorithmic}
\usepackage{algorithm}
\usepackage{array}
\usepackage{textcomp}
\usepackage{stfloats}
\usepackage{url}
\usepackage{verbatim}
\usepackage{graphicx}
\usepackage{cite}
\usepackage{amsthm,bm}
\usepackage{booktabs}
\usepackage{bm}
\usepackage{wrapfig}
\usepackage{thmtools,thm-restate}
\usepackage{graphicx}
\usepackage{subcaption}
\usepackage{float}
\usepackage{multirow}
\usepackage{placeins}

\usepackage{xcolor}

\newtheorem{theorem}{Theorem}
\newtheorem{definition}{Definition}

\newtheorem{proposition}{Proposition}
\newtheorem{lemma}{Lemma}
\newtheorem{assumption}{Assumption}

\DeclareMathOperator*{\argmin}{arg\,min}
\begin{document}

\title{Bayes-Nash Generative Privacy Protection Against Membership Inference Attacks}

\author{Tao Zhang, Rajagopal Venkatesaramani, Rajat K. De, Bradley A. Malin, Yevgeniy Vorobeychik
\thanks{T. Zhang and Y. Vorobeychik are with Washington University in St. Louis, St. Louis, MO, USA.  
Email: tz636@nyu.edu; yvorobeychik@wustl.edu

R. Venkatesaramani is with Northeastern University, Boston, MA, USA.  
Email: r.venkatesaramani@northeastern.edu

R. K. De is with the Indian Statistical Institute, Kolkata, India.  
Email: rajat@isical.ac.in

B. A. Malin is with Vanderbilt University, Nashville, TN, USA.  
Email: b.malin@vumc.org}
}

\markboth{Journal of \LaTeX\ Class Files,~Vol.~14, No.~8, August~2021}%
{Shell \MakeLowercase{\textit{et al.}}: A Sample Article Using IEEEtran.cls for IEEE Journals}


\maketitle

\begin{abstract}

Membership inference attacks (MIAs) pose significant privacy risks by determining whether individual data is in a dataset. While differential privacy (DP) mitigates these risks, it has limitations including limited resolution in expressing privacy-utility tradeoffs and intractable sensitivity calculations for tight guarantees. We propose a game-theoretic framework modeling privacy protection as a Bayesian game between defender and attacker, where privacy loss corresponds to the attacker's membership inference ability. To address strategic complexity, we represent the defender's mixed strategy as a neural network generator mapping private datasets to public representations (e.g., noisy statistics) and the attacker's strategy as a discriminator making membership claims. This \textit{general-sum Generative Adversarial Network} trains iteratively through alternating updates, yielding \textit{Bayes-Nash Generative Privacy (BNGP)} strategies. BNGP avoids worst-case privacy proofs such as sensitivity calculations, supports correlated mechanism compositions, handles heterogeneous attacker preferences. Empirical studies on sensitive dataset summary statistics show our approach significantly outperforms state-of-the-art methods by generating stronger attacks and achieving better privacy-utility tradeoffs.

\end{abstract}


\section{Introduction}

Data-driven decision-making in many fields, including healthcare, education, and finance, often relies on sensitive data, such as medical records and financial transactions.
This gives rise to privacy risks to the individuals comprising these datasets, and, in turn, both financial (due to regulatory oversight) and reputation risks to the organizations hosting and working with such datasets.
As a consequence, a large literature has emerged that aims to develop technical approaches for enabling the use and sharing of sensitive data while preserving privacy~\cite{malin2005evaluation,wan2022sociotechnical,wang2010privacy,xie2021generalized}, with differential privacy (DP) among the leading approaches to this end~\cite{alghamdi2022cactus,dwork2006differential,dwork2006calibrating,dwork2010boosting,stock2022defending}.

A particularly notable class of privacy risks that has received considerable attention in the literature involves \emph{membership inference attacks (MIAs)}, in which an attacker can effectively determine whether particular individuals are included in a given dataset~\cite{bu2021haplotype,hagestedt2020membership,niu2024survey,shokri2017membership,wu2024better}.
To the extent that membership itself is sensitive (for example, for clinical data collected from individuals with a particular sensitive disease or who have undergone a particular treatment procedure), exposure to MIA is a significant privacy concern.
MIA attacks have been especially salient in studies involving summary data releases, such as releasing histograms or summary statistics of numerical features~\cite{li2018privacy,macarthur2021workshop}, particularly when the data are high-dimensional, such as the sharing of genomic data~\cite{homer2008resolving,macarthur2021workshop,sankararaman2009genomic,venkatesaramani2023defending,yilmaz2022genomic}.
On the one hand, sharing sensitive data broadly in summary form can be of considerable value in science (such as the sharing of summary genomic data in clinical science~\cite{sankararaman2009genomic,venkatesaramani2023enabling}).
On the other hand, doing so naively is known to lead to a significant privacy risk~\cite{homer2008resolving,shringarpure2015privacy}.
While numerous approaches for mitigating such privacy risks have been proposed~\cite{liu2024mitigating,nasr2018machine}, including the application of DP~\cite{truex2019effects,ha2023differential}, these typically are either limited to particular attack models, such as likelihood-based inference (e.g.,~\cite{venkatesaramani2023defending,venkatesaramani2023enabling}), result in significant degradation of data utility, or provide only weak privacy guarantees.
For example, typical techniques that apply DP in summary data sharing provide for a relatively coarse tradeoff between data utility and privacy, so that these either lose meaningful privacy guarantees in high dimensions (e.g., setting $\epsilon$ high in $\epsilon$-DP~\cite{erlingsson2014rappor,venkatesaramani2023defending,yilmaz2022genomic}), or degrade data quality rendering it unusable.

An important fundamental challenge in applying DP-based approaches in high-dimensional settings is the computational intractability involved in obtaining sufficiently tight privacy bounds.
For example, sensitivity is an important tool in obtaining practical DP mechanisms, but is NP-hard to compute in general~\cite{xiao2008output}, while computing tight privacy bounds under composition of multiple mechanisms is \#P-complete \cite{murtagh2015complexity}.
These challenges are exacerbated in complex tasks requiring repeated dataset accesses, where optimal composition bounds for correlated mechanisms remain an open problem. Consequently, designing effective privacy-preserving mechanisms for summary data sharing in high dimensions that balance privacy and utility (such as for sharing genomic data) necessitates navigating non-trivial trade-offs between computational feasibility, theoretical tightness, and data utility, and practical performance.

To address these limitations, we introduce a non-parametric framework, termed \textit{Bayes-Nash generative privacy} (BNGP), for sharing summary statistics of high-dimensional data that achieves a pragmatic tradeoff between utility and privacy.
BNGP integrates game-theoretic and information-theoretic methods to robustly quantify privacy risks by directly formulating an optimization problem that captures the privacy-utility trade-off through a novel \textit{Bayes generative privacy risk} (BGP risk) metric.
A key advantage of BNGP over common DP approaches is that it avoids reliance on worst-case privacy proofs such as sensitivity calculations and combinatorial composition analysis.

More precisely, we model the task of balancing the privacy and utility of sharing summary data against MIAs as a game between a \textit{defender} and an \textit{attacker}.
In this game, the defender aims to optimize the privacy-utility tradeoff (with associated parameters that explicitly represent this tradeoff).
One way to formalize an attacker's utility is explicitly trading off expected TPR and FPR, where expectation is with respect to a (subjective) prior distribution over dataset membership.
We can represent this utility using a Bayes-weighted membership advantage (BWMA), a generalization of the previously proposed \emph{membership advantage (MA)} metric~\cite{yeom2018privacy}.
Proposition \ref{prop:weighted_membership_advantage} shows that BWMA is fully characterized by a simple loss function.
We assume that the defender knows the prior distribution over the datasets, whereas the attacker's prior may be subjective (and possibly incorrect).

A key limitation of such attack models, common in game-theoretic security approaches, is assuming the defender knows both the attacker's weights for trading off TPR and FPR and their subjective prior.
To address this, we propose a novel \textit{BGP risk}, defined as the minimum expected cross-entropy achievable by any attacker, independent of preferences between true and false positives.
We use negative BGP risk to quantify privacy loss and define \textit{BNGP} as the \textit{Bayes-Nash equilibrium} where the defender maps datasets to probability distributions over shared statistics, while the attacker maps observed statistics to membership inference claims.
Theorem \ref{thm:BNGP} shows that the defender's BNGP strategy (henceforth, \textit{BNGP(D)}) minimizes worst-case privacy loss independent of attacker preferences, and Section~\ref{sec:subjective_prior} demonstrates robustness to any subjective prior.
The BNGP framework satisfies post-processing (Proposition \ref{prop:post_processing}) and composition properties (Theorem \ref{prop:composition_risk}), enabling privacy optimization for complex, correlated mechanisms beyond differential privacy's independent mechanism assumptions.

At face value, even the problem of representing the space of defender and attacker strategies in this framework is intractable, let alone one of computing an equilibrium.
We address this by representing defender strategies as a parametric differentiable (e.g., neural network) generator, while the attacker's strategy is represented as a probabilistic classifier mapping target individuals and observed summary statistics to a membership inference decision.
We train these jointly through alternating gradient updates. This approach resembles generative adversarial networks (GANs) but differs crucially in being general-sum rather than zero-sum, leading us to term it a \textit{general-sum GAN} with the defender as generator. 
Our analysis shows that approximation error decreases at a quantifiable rate as the attacker's neural network capacity increases, guaranteeing that sufficient network capacity enables learned attacker behavior to accurately approximate true worst-case risk, ensuring our approach's robustness.

We demonstrate the efficacy of the proposed approach in the context of sharing summary statistics on high-dimensional datasets with binary attributes. 
In particular, we compared \emph{BNGP(D)} with state-of-the-art baselines including $\epsilon$-DP, as well as the mixed-strategy MIA obtained in the process (\emph{BNGP(A)}) to associated state-of-the-art baselines.
Our experiments show that (a) \emph{BNGP(A)} is indeed the strongest MIA approach in this setting in practice, outperforming all baselines in MIA efficacy, and (b) \emph{BNGP(D)} outperforms baseline defenses in terms of privacy-utility tradeoff (in particular, privacy for a given target utility) when faced with \emph{BNGP(A)} attacks.

We summarize our contributions:
\begin{itemize}
    \item[(1)] In Section \ref{sec:Bayesian_game}, we model the optimal privacy-utility trade-off as a Bayesian game, where the defender can customize the objective function under a mild assumption. We further extend the conventional membership advantage (MA) metric to a Bayes-weighted MA, which accounts for the attacker’s arbitrary preferences over the true positive rate (TPR) and false positive rate (FPR) while also incorporating any subjective prior knowledge about the true membership status. 
    
    \item[(2)] We introduce a non-parametric framework, termed Bayes-Nash generative privacy (BNGP), which integrates game-theoretic and information-theoretic methods by formulating an optimization problem based on a novel Bayes generative privacy risk (BGP risk) metric. This formulation avoids reliance on worst-case privacy proofs—such as sensitivity calculations and combinatorial composition analysis required by DP. Moreover, by treating the entire underlying data process as a black-box oracle, the BNGP framework enables efficient optimization even for complex systems.
    We show that BNGP satisfies the post-processing property and composes gracefully even across interdependent mechanisms. 
    
    \item[(3)] To evaluate our approach, we conduct experiments on high-dimensional genomic datasets with binary attributes. Our results show that the BNGP outperforms existing membership inference attack baselines, demonstrating that BNGP provides robust privacy protection for privacy-preserving data sharing.    
\end{itemize}

All theoretical proofs are provided in the Appendix of the full version \cite{zhang2025bayes}.

\subsection{Related Work}

\textbf{Quantitative Notions of Privacy Leakage }
Quantitative notions of privacy leakage have been extensively studied in various contexts which provides mathematically rigorous frameworks for measuring the amount of sensitive information that may be inferred by attackers.
Differential privacy \cite{dwork2006calibrating,dwork2006differential} and its variants \cite{bun2016concentrated,dwork2016concentrated,mironov2017renyi,bun2018composable} formalize the privacy leakage using various parameterized statistical divergence metrics.
For example, R\'enyi differential privacy (RDP) \cite{mironov2017renyi} generalizes the standard pure DP and quantifies the privacy leakage through the use of R\'enyi divergence.
Information-theoretic measures, such as mutual information \cite{chatzikokolakis2010statistical,cuff2016differential}, f-divergence \cite{xiao2023pac}, and Fisher information \cite{farokhi2017fisher,hannun2021measuring,guo2022bounding}, provide alternative ways to quantify and characterize privacy loss. 
Probably Approximately Correct (PAC) Privacy \cite{xiao2023pac,sridhar2025pac} and its refinement Residual-PAC Privacy \cite{zhang2025breaking} further advance this direction by offering instance-based, simulation-driven guarantees that automatically calibrate noise while overcoming the Gaussian-bound limitation.
Empirical measurements are also widely studied \cite{shokri2017membership,yeom2018privacy,nasr2021adversary,stock2022defending} that quantify the actual privacy guarantees or leakage under certain privacy protection methods.

\textbf{Privacy-Utility Trade-off }
Balancing the trade-off between privacy and utility is a central challenge in designing privacy-preserving mechanisms. This balance is often modeled as an optimization problem \cite{lebanon2009beyond,sankar2013utility,lopuhaa2024mechanisms,ghosh2009universally,gupte2010universally,geng2020tight,du2012privacy,alghamdi2022cactus,goseling2022robust}. For instance, \cite{ghosh2009universally} formulated a loss-minimizing problem constrained by differential privacy and demonstrated that the geometric mechanism is universally optimal in Bayesian settings. Similarly, optimization problems can be framed with utility constraints \cite{lebanon2009beyond,alghamdi2022cactus}. Moreover, \cite{gupte2010universally} modeled the trade-off as a zero-sum game, where the privacy mechanism maximizes privacy while information consumers minimize their worst-case loss using side information.

\textbf{GAN in Privacy }
The use of generative adversarial networks (GANs) for privacy protection has gained increasing attention in recent years.
\cite{huang2018generative} introduced generative adversarial privacy (GAP), which frames privacy as a game between a generator that creates utility-preserving data while obfuscating sensitive information and a discriminator that attempts to identify private data. The objective is training models that achieve both high utility and resilience to powerful inference attacks.
Similar efforts include compressive adversarial privacy (CAP), which compresses data before adversarial training \cite{chen2018understanding}, adversarial regularization methods that adjust training to reduce information leakage \cite{nasr2018machine}, and PATE-GAN, which combines the PATE framework with GANs for differentially private synthetic data generation \cite{jordon2018pate}.
Other works in this line includes privacy-preserving adversarial networks \cite{tripathy2019privacy}, reconstructive adversarial network \cite{liu2019privacy}, 
and federated GAN \cite{rasouli2020fedgan}.
Adversarial training has also been applied to defend against MIAs specifically. For example, \cite{li2021membership} explored methods where models are trained alongside adversaries attempting MIAs, which enables the models to learn representations that are less susceptible to such attacks. 
There are also works using GAN to perform attacks, where the generator represents the attacker's strategy \cite{baluja2017adversarial,hitaj2017deep,zhao2018generating,liu2019performing,hayes2019logan}.

\section{Preliminaries}\label{sec:preliminaries}

\subsection{Differential Privacy}

The core challenge of privacy-preserving data analyses can be formulated as the following problem: Let $D$ represent a dataset that is processed by a (potentially randomized) mechanism $\mathcal{M}: \mathcal{D} \to \mathcal{X}$, which generates an output $X = \mathcal{M}(D)$ that is publicly available. 
An output could be summary statistics \cite{sankararaman2009genomic,dwork2015robust}, any model learned from the dataset \cite{abadi2016deep,shokri2017membership}, or other information or signals such as network traffic when processing the dataset \cite{chen2010side}. 
Here, $\mathcal{D}$ and $\mathcal{X}$ denote the input and output spaces, respectively.
DP ensures that any two \textit{adjacent datasets} $D, D' \in \mathcal{D}$, differing by a single entry ($D \simeq D'$), produce nearly indistinguishable outputs, limiting the information inferable about individual data records.
A randomized mechanism $\mathcal{M}: \mathcal{D} \to \mathcal{X}$ is \textit{$(\epsilon, \delta)$-differentially private} ($(\epsilon, \delta)$-DP) with $\epsilon \geq 0$ and $\delta \in [0,1]$ if for all $D \simeq D'$, $\mathcal{W}\subset \mathcal{Y}$,
\begin{equation}\label{eq:standard_DP_def}
    \Pr\left[\mathcal{M}(D) \in \mathcal{W}\right]  \leq e^{\epsilon} \Pr\left[\mathcal{M}(D') \in \mathcal{W}\right] + \delta.
\end{equation}
This definition constrains the statistical distance between the probability distributions $\mathcal{M}(B)$ and $\mathcal{M}(B')$, ensuring that their outputs are nearly indistinguishable.

\paragraph{Sensitivity}
A critical parameter in designing differentially private mechanisms is \textit{sensitivity}, which quantifies the maximum influence a single individual’s data can have on the output of a query. Formally, for a function $f: \mathcal{D} \to \mathcal{X}$, the \textit{global sensitivity} $\mathtt{Sens}(f)$ is defined as the maximum $L_{1}$-norm (or $L_{2}$-norm, depending on the mechanism) difference between the outputs of $f$ over all pairs of adjacent datasets $D \simeq D' $:  
\[
\mathtt{Sens}(f) = \max_{D \simeq D'} \|f(D) - f(D')\|.
\]  
For example, a counting query (e.g., "how many individuals have a specific attribute?") has sensitivity $\mathtt{Sens}(f) = 1$, since adding or removing one individual changes the count by at most 1. Sensitivity directly determines the magnitude of noise required to privatize the query. Mechanisms like the Laplace mechanism \cite{dwork2006calibrating} add noise scaled to $\mathtt{Sens}(f) = 1 / \epsilon$, ensuring that the output distribution’s privacy loss adheres to the $\epsilon$-DP guarantee. Lower sensitivity enables stronger privacy (smaller $\epsilon$) for the same utility, making sensitivity analysis a cornerstone of DP algorithm design.  
However, sensitivity calculation is in general NP-hard \cite{xiao2008output}.

\paragraph{Composition Properties}
In practice, we want to use the data more than once.
Thus, it is important to characterize how privacy guarantees degrade when multiple DP mechanisms are applied to the same (or related) dataset.
DP satisfies graceful composition properties.
That is, DP is closed under composition: composition keeps privacy intact, while gradually eroding its strength.
The \textit{basic composition theorem} states that the privacy parameters $\epsilon$ and $\delta$ accumulate linearly under sequential composition: executing $n$ mechanisms with $(\epsilon_i, \delta_i)$-DP guarantees yields $\left(\sum_{i=1}^k \epsilon_i, \sum_{i=1}^{n} \delta_i\right)$-DP total privacy loss.
The textit{advanced composition theorem} \cite{dwork2010boosting} provides tighter bounds: for $k$-fold composition of $(\epsilon, \delta)$-DP mechanisms, total privacy loss is bounded by $(\epsilon', \delta')$-DP, where $\epsilon' = \epsilon \sqrt{2k \log(1/\delta'')} + k \epsilon (e^\epsilon - 1), \quad \delta' = k\delta + \delta''$ for any $\delta'' > 0 $.
Kariouz et al. \cite{kairouz2015composition} obtain optimal composition bounds for homogeneous cases, while Murtagh and Vadhan \cite{murtagh2015complexity} provide optimal bounds for heterogeneous cases and show that it is, in general, \#P-complete to compute the bound.
These composition properties enable modular design of complex DP workflows while maintaining rigorous privacy accounting.

\subsection{Membership Inference Attacks}\label{sec:MIA}

Membership inference attacks (MIAs) aim to determine whether a specific data point was included in a model’s training dataset. 
Consider a population $\mathcal{K} = [K]$ of $K$ individuals, where each individual $k$ has an associated data point $d_k$ (e.g., a feature vector). 
Let $Z=\{d_{k}\}_{k\in\mathcal{K}}\in \mathcal{Z}$ be the \textit{population dataset}, where $\mathcal{Z}$ is the set of population datasets.
A binary \textit{membership vector} $b = (b_1, b_2, \dots, b_K) \in W \equiv \{0,1\}^K$ indicates dataset inclusion, where $b_k = 1$ if $d_k$ is included and $b_k = 0$ otherwise.
Let $\mathcal{D}=\mathcal{P}(Z)\backslash\{\emptyset\}$ be the \textit{power set} of $Z$ excluding the empty set.
The dataset is thus defined as $D = \{ d_k \mid b_k = 1, k \in \mathcal{K} \}\in\mathcal{D}$.
Membership selection follows a prior distribution $\theta(\cdot) \in \Delta(W)$, governing the probability of each possible membership vector.
We assume that each individual's membership information is independent of the others. 
Consider a (potentially randomized) mechanism $f(D)$ that takes $D$ as input and produces an output $x \in \mathcal{X}$, where $\mathcal{X}$ is the set of possible outputs.

\paragraph{Example: Summary Statistic Sharing}
Consider a population $\mathcal{K} = [K]$ of $K$ individuals, where each individual's data is represented by a binary vector $d_k = (d_{kj})_{j \in Q}$, with $d_{kj} \in \{0,1\}$ specifying the binary value of the attribute at position $j$. The set $Q$ represents all attributes under consideration, such as genomic positions or other features. 
The dataset $D$ consists only of data points from individuals $k\in \mathcal{K}$ who satisfy $b_{k}=1$.
The data-sharing mechanism $f(D) = x = (x_1, \dots, x_{|Q|}) \in \mathcal{X} = [0,1]^{|Q|}$ computes summary statistics $x$, where each $x_j$ represents the fraction of included individuals with $d_{kj} = 1$ for attribute $j$: $x_j = \frac{\sum_{k\in\mathcal{K}} b_k d_{kj}}{\sum_{k\in\mathcal{K}} b_k}$.
For example, in genomic data, $d_k$ may represent single-nucleotide variants (SNVs), where each $d_{kj}$ indicates the presence of an alternate allele at SNV $j$ for individual $k$. The summary statistic in this case, known as the alternate allele frequency (AAF), measures the fraction of included individuals carrying the alternate allele at each SNV.

An MIA model is a (potentially randomized) mechanism $\mathcal{A}(d_k, x)$ $\in \{0,1\}$ that predicts whether an individual $k$ was included in the dataset, given their data point $d_k$ and the output $x$ of the mechanism $f$.
The \textit{membership advantage} (MA) \cite{yeom2018privacy} is a standard measure of MIA performance, defined for each $k \in \mathcal{K}$ as:
\begin{equation}\label{eq:membership_advantage}
    \begin{aligned}
        &\mathtt{Adv}_{k}(\mathcal{A}) \\&\equiv \Pr[\mathcal{A}(d_k, X) = 1 \mid b_k = 1] - \Pr[\mathcal{A}(d_k, X) = 1 \mid b_k = 0]\\
        &=2\Pr[\mathcal{A}(d_k, X) = b_k] - 1.
    \end{aligned}
\end{equation}
In other words, $\mathtt{Adv}_k(\mathcal{A})$ quantifies the difference between the model’s \textit{true positive rate (TPR)} and \textit{false positive rate (FPR)}.
Other metrics for evaluating MIA performance include accuracy \cite{shokri2017membership}, area under the curve (AUC) \cite{chen2020gan}, mutual membership information leakage \cite{farokhi2020modelling}, and privacy-leakage disparity \cite{zhong2022understanding}. For a comprehensive review, see \cite{niu2024survey}.

\paragraph{DP for MIA}
MIAs provide a direct theoretical framework for characterizing the privacy risk in DP. The core definition of DP ensures that an adversary cannot distinguish whether a specific individual’s data is present in a dataset by observing the output of a DP mechanism. Since MIAs instantiate precisely this distinguishability task, they align tightly with DP’s fundamental guarantees and can serve as provable bounds on privacy leakage.
Suppose that we have a $(\epsilon, \delta)$-DP mechanism $\mathcal{M}: \mathcal{D} \mapsto \mathcal{X}$, where $\mathcal{M}$ is the privacy-preserving version of the original mechanism $f: \mathcal{D}\mapsto \mathcal{X}$.
Then, \cite{humphries2023investigating} shows that the adversary's membership advantage is upper bounded by
\[
\begin{aligned}
    \mathtt{Adv}_{k}(\mathcal{A})&=2\Pr[\mathcal{A}(d_k, \mathcal{M}(D)) = b_k] - 1 \\
    &\leq \left(e^{\epsilon} -1 + 2\delta\right)/\left(e^{\epsilon} +1\right),
\end{aligned}
\]
for all $k\in \mathcal{K}$.
The bound asymptotically approaches zero as $\delta, \epsilon \rightarrow 0$; it increases monotonically with $\delta$ for fixed $\epsilon\geq 0$ and with $\epsilon$ for fixed $\delta\in[0,1)$.
This bound holds in the worst-case scenario, applying even if the attacker has unbounded computational power.

\section{Privacy Protection Against MIA as a Game}\label{sec:Bayesian_game}

We formalize the interaction between two agents: the \textit{defender}, who acts as the data curator of dataset $D$, and the \textit{attacker}, an adversary attempting to perform MIAs. 
The defender aims to obtain an optimal trade-off between privacy and utility.
Conversely, the attacker seeks to maximize inference accuracy while potentially operating under computational or resource constraints.

\subsection{Privacy Strategy}

To protect membership privacy, the defender randomizes the mechanism $f$ via \textit{noise perturbation}. 
Noise perturbation is one of the most common techniques used to enhance privacy protection, such as in differential privacy, by adding controlled random noise to data, query responses, or model parameters to obscure the contribution of any single individual.

Specifically, let $g: W \mapsto \Delta(\Xi)$ denote the \textit{privacy strategy}, where $g(\xi|b)$ specifies the probability distribution over \textit{noise} $\xi \in \Xi$, conditioning on the membership vector $b\in W$.
The privacy strategy may also be independent of $b$, i.e., $g(\cdot) \in \Delta(\Xi)$.
The randomized version of the original mechanism $f$ is represented as the mechanism $\mathcal{M}(\cdot; g): \mathcal{D}\mapsto \mathcal{X}$.
In addition, let $\rho_{g}: \mathcal{D} \mapsto \Delta(\mathcal{X})$ be the underlying density function given $g$ and $f$.

When an output $x = \mathcal{M}(D; g)$ is realized with $g$ drawing a noise $\xi$, we denote it as $x = \mathtt{M}(D,\xi)$. In output perturbation, $\xi$ is added to the output $\hat{x} = f(D)$, and the publicly released output is $x = \mathtt{M}(D,\xi) = \mathtt{R}(\hat{x} + \xi)$, where $\mathtt{R}(\cdot)$ ensures the perturbed $x$ remains within the valid range $\mathcal{X}$. For example, as described in Section \ref{sec:preliminaries}, when $\hat{x}$ represents frequencies, the formulation $x = \mathtt{R}(\hat{x} + \xi) \equiv \mathrm{Clip}_{[0,1]}(\hat{x} + \xi)$ ensures $x \in [0,1]^{|Q|}$.

\subsection{Attack Model and Bayes-Weighted Membership Advantage}

The attacker conducts a membership inference attack (MIA) to infer the true membership vector. 
A trivial attack in MIA occurs when the attacker always assumes that target individuals are present in the dataset. While this leads to a high TPR, it comes at the cost of an extremely high FPR, making it an impractical strategy. Such an attack ignores the operational costs associated with false positives and fails to exploit actual privacy vulnerabilities. As a result, considering trivial worst-case scenarios can lead to misleading privacy risk assessments and suboptimal privacy-utility trade-offs.

To constitute a meaningful privacy compromise, an attack must seek to outperform random guessing and be evaluated in terms of both effectiveness and operational costs. In practical applications (e.g., personalized medicine or targeted marketing), an attacker does not merely seek to maximize true positives (i.e., correctly inferring membership) but must also account for the costs of acting on false positives. A reckless attack that assumes universal membership may lead to wasted resources, inefficiencies, or even detection.

The attacker may possess external knowledge represented as \textit{subjective prior beliefs} about the true membership $b \in W$, denoted by $\sigma(\cdot) \in \Delta(W)$. The attack model (possibly randomized), $\mathcal{A}_{h}(x) = \{\mathcal{A}^{k}_{h}(d_k, x)\}_{k \in \mathcal{K}}$, consists of individual decisions $\mathcal{A}^{k}_{h}(d_k, x) \in \{0,1\}$, where $h:\mathcal{X}\mapsto\Delta(W)$ captures the intrinsic randomness of the attack model that is independent of $\sigma$.
Let $s=(s_{k})_{k\in \mathcal{K}}\in W$ denote the inference output of $\mathcal{A}_{h}(x)$ for any $x\in\mathcal{X}$.
Then, $\Pr[\mathcal{A}_{h}(x)=s] = h(s|x)$, for all $s\in W$, $x\in\mathcal{X}$.

To capture the attacker's trade-off, we extend the MA in (\ref{eq:membership_advantage}) to the \textit{$(\sigma,\gamma)$-Bayes-weighted membership advantage} ($(\sigma,\gamma)$-BWMA) by introducing a coefficient $0<\gamma\leq 1$ to reflect the attacker's trade-off between TPR and FPR. That is, given the privacy strategy $g$, $(\sigma,\gamma)$-BWMA of the attack model $\mathcal{A}_{h}$ is defined as:
\begin{equation}\label{eq:BW_MA}
    \begin{aligned}
    \mathtt{Adv}^{\gamma}_{\sigma}\left(\mathcal{A}_{h}, g\right)\equiv (1-\gamma)\mathtt{TPR}_{\sigma}\left(\mathcal{A}_{h}, g\right)
        - \gamma\mathtt{FPR}_{\sigma}\left(\mathcal{A}_{h}, g\right),
    \end{aligned}
\end{equation}
where $\mathtt{TPR}_{\sigma}\left(\mathcal{A}_{h}, g\right)$ and $\mathtt{FPR}_{\sigma}\left(\mathcal{A}_{h}, g\right)$ are the TPR and the FPR, given by
\[
\begin{aligned}
    &\mathtt{TPR}_{\sigma}\left(\mathcal{A}_{h}, g\right) \\&\equiv \sum\nolimits_{k\in \mathcal{K}, b_{-k}}\Pr\left[\mathcal{A}^{k}_{h}(d_{k}, x)=1\middle| b_{k}=1;g\right]\sigma(b_{k}=1, b_{-k}),
\end{aligned}
\]
\[
\begin{aligned}
    &\mathtt{FPR}_{\sigma}\left(\mathcal{A}_{h}, g\right)\\
    &\equiv \sum\nolimits_{k\in \mathcal{K}, b_{-k}}\Pr\left[\mathcal{A}^{k}_{h}(d_{k}, x)=1\middle| b_{k}=0;g\right]\sigma(b_{k}=0, b_{-k}).
\end{aligned}
\]
Decreasing $\gamma$ indicates the attacker's stronger preference for TPR, while increasing $\gamma$ reflects a greater desire to minimize FPR.
When $\gamma=0.5$, the attacker values TPR and FPR equally.

\subsection{Game Formulation}

We model the interaction between the defender and the attacker as a game.

\paragraph{Attacker's Objective Function}
Given $\gamma$ and $\sigma$, the attacker aims to obtain $\mathcal{A}_{h}$ that maximizes $\mathtt{Adv}^{\gamma}_{\sigma}\left(\mathcal{A}_{h}, g\right)$.
First, we consider a simple loss function for the attacker:
\begin{equation}\label{eq:attacker_basic_loss}
    \ell^{\gamma}_{\texttt{att}}(s,b) \equiv -v(s,b) +\gamma c(s),
\end{equation}
where $v(s,b)\equiv\sum\nolimits_{k\in \mathcal{K}} s_{k}b_{k}$ captures the sum of true positives, and $c(s) \equiv \sum\nolimits_{k\in \mathcal{K}} s_{k}$ captures the operational costs to post-process positive inference outcomes (i.e., for $k\in\mathcal{K}$ such that $s_{k}=1$).
Maximizing $v(s,b)$ reflects maximizing true positives, while minimizing $c(s)$ reflects minimizing the operational costs.

Given a privacy strategy $g$ (and the induced $\rho$), prior $\sigma$, and the (probability of the) attack model $h$, the expected loss is defined as:
\begin{equation}\label{eq:attacker_original_loss}
    \mathcal{L}^{\gamma,\sigma}_{\texttt{att}}(h, g)\equiv \sum\nolimits_{s,b}\int_{\mathcal{X}}\ell^{\gamma}_{\texttt{att}}(s,b)h(s|x)\rho(x|b)dx \sigma(b).
\end{equation}

\begin{proposition}\label{prop:weighted_membership_advantage} 
For any $g$, $\sigma$, $h$, and $0 < \gamma \leq 1$, we have $\mathcal{L}^{\gamma,\sigma}_{\texttt{att}}(h,g) = -\mathtt{Adv}^{\gamma}_{\sigma}\left(\mathcal{A}_{h}, g\right)$.
\end{proposition}

Proposition \ref{prop:weighted_membership_advantage} establishes that for $0 < \gamma \leq 1$, the attacker's optimal strategy simultaneously minimizes $\mathcal{L}^{\gamma,\sigma}_{\texttt{att}}(h, g)$ and maximizes $\mathtt{Adv}^{\gamma}_{\sigma}(\mathcal{A}_{h}, g)$ for any given $g$. 
Here, $s$ represents the attacker's pure strategy and $h$ corresponds to the attacker's mixed strategy.
Henceforth, we use $h$ to represent the attack model $\mathcal{A}_{h}$, unless otherwise stated.
This equivalence streamlines the defender-attacker interaction by framing it as a Bayesian game, where the attacker is a rational player optimizing expected payoff, thereby enabling the use of game-theoretic solution concepts to analyze the defender-attacker interaction.

\paragraph{Defender's Objective Function}
Given a population dataset $Z \in \mathcal{Z}$, the defender seeks to optimally balance the \textit{privacy-utility trade-off} by formulating and solving an optimization problem. The specific formulation of this optimization problem depends inherently on the context, reflecting domain-specific priorities, operational constraints, and sociotechnical assumptions.

Privacy loss and utility loss are commonly considered competing objectives. However, their mathematical representation is not universally standardized; instead, it is guided by implicit preferences, domain-specific privacy considerations, and operational interpretations of privacy and utility in a given scenario. As such, there is no single, universally optimal formulation suitable for every situation.

Formally, let $h$ represent the attacker’s strategy or attack model, and $g$ represent the defender’s privacy protection strategy. Define $\mathtt{PrivacyL}(h, g)$ as the \textit{privacy loss} experienced by the defender under attack strategy $h$, and $\mathtt{UtilityL}(g|Z)$ as the \textit{utility loss} resulting from the defender’s privacy protection strategy $g$ when applied to dataset $Z$.
Additionally, let $\mathcal{L}_{\texttt{def}}(h,g)$ denote the defender’s objective function. 

Below are three widely used ways to encode a privacy-utility trade-off. 

\paragraph{Trade-off with Privacy Budget}
Given a fixed attack strategy $h$, suppose the defender prioritizes minimizing the utility loss while ensuring the privacy loss does not exceed a predefined threshold known as the privacy budget ($\textup{PB} \in \mathbb{R}$). The defender’s objective function in this scenario can be expressed as:
\begin{equation}\label{eq:trade_off_PB}
    \mathcal{L}_{\texttt{def}}(h,g)=
    \begin{cases}
        \mathtt{UtilityL}(g|Z), & \text{if } \mathtt{PrivacyL}(h,g) \leq \textup{PB},\\
        \infty, & \text{otherwise}.
    \end{cases}
\end{equation}
In this formulation, strategies exceeding the privacy budget become infeasible and are penalized infinitely.

\paragraph{Trade-off with Utility Budget}
Given an attack model $h$, the defender might instead prioritize minimizing privacy loss, while constraining the utility loss within an acceptable limit, termed the utility budget ($\textup{UB} \in \mathbb{R}$). This scenario is represented by the following objective function:
\begin{equation}\label{eq:trade_off_UB}
    \mathcal{L}_{\texttt{def}}(h,g)=
    \begin{cases}
        \mathtt{PrivacyL}(h,g), & \text{if } \mathtt{UtilityL}(g|Z) \leq \textup{UB},\\
        \infty, & \text{otherwise}.
    \end{cases}
\end{equation}
Here, strategies exceeding the utility budget become infeasible.

\paragraph{Trade-off with Preference}
Alternatively, the defender may explicitly balance privacy and utility losses within a single combined objective function, using a weighting factor $\kappa > 0$ that captures the defender’s preference over this trade-off:
\begin{equation}\label{eq:additive_obj}
    \mathcal{L}_{\texttt{def}}(h,g) = \mathtt{PrivacyL}(h,g) + \kappa \mathtt{UtilityL}(g|Z).
\end{equation}
This additive form allows the privacy loss to be interpreted directly as a monetary or quantifiable resource that can be traded against utility loss.

Our model is general and flexible enough to accommodate any defender-defined objective functions that align with their operational and strategic needs. However, we introduce the following minimal assumption to ensure the logical consistency and interpretability of the privacy loss metric:

\begin{assumption}\label{assp:trade-off}
For any fixed privacy protection strategy $g$, the privacy loss $\mathtt{PrivacyL}(h,g)$ is monotonically increasing with respect to either $\mathtt{TPR}_{\theta}(h,g)$ or $\mathtt{Adv}^{0.5}_{\theta}(h,g)$. In other words, privacy loss should not decrease as the attacker’s success or advantage in distinguishing sensitive information increases.
\end{assumption}

This assumption ensures that the defender’s notion of privacy loss behaves intuitively: if the attacker becomes more successful in inferring sensitive information—either by achieving a higher true positive rate ($\mathtt{TPR}_{\theta}(h,g)$) or gaining more advantage over random guessing ($\mathtt{Adv}^{0.5}_{\theta}(h,g)$)—then the perceived privacy loss should not decrease. That is, stronger attacks should correspond to greater privacy degradation. This property aligns with the defender’s rational expectation that a more effective adversary implies a more severe privacy breach, and it helps rule out pathological or unintuitive definitions of privacy loss that might assign lower loss to worse outcomes.

The utility loss can focus on the utility generated by the post-processing of the mechanism output, such as the value to data users of published statistics \cite{schmutte2022information}.
A common way to model the utility loss is by the deviation of $x=\mathcal{M}(D;g)$ from the unperturbed output $\hat{x}=f(D)$. 
Specifically, let $\ell_{U}:\mathbb{R}_{+}\mapsto \mathbb{R}_{+}$ be an increasing, differentiable function, and let $\|\cdot\|_{\mathtt{p}}$ be a norm on $\mathcal{X}$, for $\mathtt{p}\geq 1$.
Then, $\mathtt{UtilityL}\left(g\right)$ can be defined by the expectation of $\ell_{U}(\|x-\hat{x}\|_{\mathtt{p}})$.
Hence, one example of the defender's loss is
\begin{equation}\label{eq:defender_loss}
    \ell_{\texttt{def}}(b, s, x) \equiv v(s, b) + \kappa \ell_{U}(\|x - f(D)\|_{\mathtt{p}}).
\end{equation}
Given any $g$ (and the induced $\rho_{g}$) and $h$, the defender's expected loss $\mathcal{L}_{\texttt{def}}$ in the format (\ref{eq:additive_obj}) is then given by
\begin{equation}
\mathcal{L}_{\texttt{def}}\left(h,g\right)\equiv\sum\nolimits_{s,b}\int_{x} \ell_{\texttt{def}}(b, s, x)h(s|x)\rho_{g}(x|b)\theta(b)dx.
\end{equation}

The interaction between the defender and attacker is modeled as a game, with each optimizing their strategy. A \textit{$\sigma$-Bayesian Nash Equilibrium} represents the point where neither can unilaterally improve their outcome.

\begin{definition}[$\sigma$-Bayes Nash Equilibrium] 
Given any $\sigma$ and $0<\gamma\leq 1$, a profile $\left<g^{*}, h^{*}\right>$ is a \textup{$\sigma$-Bayesian Nash Equilibrium ($\sigma$-BNE)} if 
\begin{equation}\label{eq:BNE_def}
    \begin{aligned}
    g^{*}\in\argmin\nolimits_{g} \mathcal{L}_{\texttt{def}}\left(h^{*},g\right) \textup{ and } h^{*}\in\argmin\nolimits_{h} \mathcal{L}^{\gamma,\sigma}_{\texttt{att}}(h, g^{*}).
\end{aligned}
\end{equation}
\end{definition}

\subsection{GAN-Based Game}

We propose representing the defender’s and attacker’s strategies, $g$ and $h$, respectively, using neural networks due to their flexibility and expressiveness. 
In the infinite-capacity limit, neural nets can approximate any distribution (by universal approximation).

This perspective naturally suggests a Generative Adversarial Network (GAN)-like framework, which we term the \textit{general-sum GAN}. In this framework, the defender’s privacy strategy $g$ is implemented as a neural network \textit{generator}, denoted as $G(b, r)$. This generator takes two inputs: the true membership vector $b$, and an auxiliary random vector $r$. The output of $G$ is a noise vector $\xi$, which is then used for privacy protection. The auxiliary vector $r$, with dimension $q$, is drawn from a uniform distribution, i.e., $r \sim \mathcal{U}([0,1])$, ensuring randomness in the defender’s strategy.

The attacker’s strategy is modeled as a neural network \textit{discriminator}, denoted as $H: \mathcal{X} \mapsto W$. The discriminator takes an input $x = \mathtt{M}(D,G(b, r))$ (i.e., $x=\mathcal{M}(D;g)$ when noise $G(b, r)$ is generated), representing the data transformed by the defender’s privacy strategy, and outputs inference vector $s \in W$ about the membership.

Using these neural network representations, we can explicitly rewrite the expected loss functions for both the attacker and the defender as follows:
\begin{align*}
&\widetilde{\mathcal{L}}^{\gamma,\sigma}_{\texttt{att}}(H,G) \equiv \mathbb{E}^{r \sim \mathcal{U}}_{b \sim \sigma}\left[\ell^{\gamma}_{\texttt{att}}\left(H(\mathtt{M}(D,G(b,r))),b\right)\right],\\
&\widetilde{\mathcal{L}}_{\texttt{def}}(H,G) \\
    &\equiv \mathbb{E}^{r \sim \mathcal{U}}_{b \sim \sigma}\left[\ell_{\texttt{def}}\left(b, H(\mathtt{M}(D,G(b,r))),\mathtt{M}(D,G(b,r))\right)\right],
\end{align*}
where $\ell^{\gamma}_{\texttt{att}}$ and $\ell_{\texttt{def}}$ correspond to the attacker’s and defender’s loss functions, previously defined in equations (\ref{eq:attacker_basic_loss}) and (\ref{eq:defender_loss}).

Consequently, the defender and attacker engage in the following game:
\begin{equation}\label{eq:NN_BNE}
    \begin{aligned}
    &G^{*}\in\arg\min\nolimits_{G} \widetilde{\mathcal{L}}_{\texttt{def}}(H^{*},G), \\
    &H^{*}\in \arg\min\nolimits_{H}\widetilde{\mathcal{L}}^{\gamma,\sigma}_{\texttt{att}}(H,G^{*}).
    \end{aligned}
\end{equation}
This defines the neural-network $\sigma$-Bayes-Nash equilibrium (BNE): $G^*$ and $H^*$ are best responses to each other.
Here, the neural networks $G$ and $H$ implicitly encode probability distributions for mixed strategies $g$ and $h$, respectively.
For notational clarity and consistency, we simply replace $g \rightarrow G$ and $h \rightarrow H$.
Additionally, if the original objective function $\mathcal{L}_{\texttt{def}}(h, g)$ satisfies Assumption \ref{assp:trade-off}, this property naturally extends to the neural network formulation $\widetilde{\mathcal{L}}_{\texttt{def}}(H, G)$ as well, preserving the intuitive and analytical integrity of the privacy loss metric.

\section{Bayes-Nash Generative Privacy Strategy}\label{sec:BNGP}

In this section, we propose the \textit{Bayes-Nash generative privacy} (BNGP) framework.
In the GAN-based game (\ref{eq:NN_BNE}), the outputs of $G$ and $H$ are the samples of noise and inference.
In contrast, in the BNGP framework, we consider a different construction of the neural network $H$ which takes inputs as $x$ and outputs a vector of soft scores $p=(p_{k})_{k\in\mathcal{K}}$, where each $p_{k}\in[0,1]$ is the probability of $s_{k}=1$. 
A final inference $s \in W$ is then obtained by thresholding each coordinate of $H(x)$.

The BNGP framework considers the \textit{cross-entropy loss} (CEL) for the attacker:
\begin{equation*}
    \begin{aligned}
        \ell_{\texttt{CEL}}\left(p, b\right)\equiv -\sum\nolimits_{k\in\mathcal{K}} \left[b_{k}\log(p_{k}) + (1-b_{k})\log(1-p_{k})\right].
    \end{aligned}
\end{equation*}
Thus, the attacker's \textit{expected CEL} is given by:
\begin{equation}\label{eq:BCEL}
    \begin{aligned}
        \mathcal{L}^{\sigma}_{\texttt{CEL}}(H,G)
        \equiv \mathbb{E}^{r\sim \mathcal{U}}_{b\sim \sigma}\left[\ell_{\texttt{CEL}}\left(H\left(\mathtt{M}\left(D, G(b, r)\right)\right),b\right)\right].
    \end{aligned}
\end{equation}
Additionally, consider the defender's privacy loss in $\widehat{\mathcal{L}}^{\sigma}_{\texttt{def}}(H, G)$ formulated by (\ref{eq:trade_off_PB})-(\ref{eq:additive_obj}) is given by 
\[
\mathtt{PrivacyL}\left(H, G\right) = - \mathcal{L}^{\sigma}_{\texttt{CEL}}(H,G).
\]
That is, we take the negative cross-entropy as the defender’s privacy loss.
We define the attacker's strategy that minimizes $\mathcal{L}^{\sigma}_{\texttt{CEL}}(H,G)$ for a given $G$ as the \textit{Bayes generative privacy response}.

\begin{definition}[Bayes Generative Privacy Response and Risk]\label{def:BGP_response_risk}
Given any $G$, the \textup{$\sigma$-Bayes generative privacy response ($\sigma$-BGP response)} to $G$ is defined as 
\[
H^{*}\in\argmin\nolimits_{H}\mathcal{L}^{\sigma}_{\texttt{CEL}}(H,G).
\]
The expected loss $\mathcal{L}^{\sigma}_{\texttt{CEL}}(H^{*},G)$ is the \textup{$\sigma$-BGP risk.}
\end{definition}

The BNGP strategy is the equilibrium strategy of the defender that minimizes $\widehat{\mathcal{L}}^{\sigma}_{\texttt{def}}(H^{*}, G)$ where $H^{*}$ is a $\sigma$-BGP response.

\begin{definition}[Bayes-Nash Generative Privacy Strategy]
    The model $G^{*}$ is a \textup{$\sigma$-Bayes-Nash generative privacy ($\sigma$-BNGP) strategy} if 
    \begin{itemize}
        \item[(1)] $G^{*}\in\argmin\nolimits_{G}\widehat{\mathcal{L}}^{\sigma}_{\texttt{def}}(H^{*}, G)$,
        \item[(2)] $H^{*}\in\argmin\nolimits_{H}\mathcal{L}^{\sigma}_{\texttt{CEL}}(H,G^{*})$.
    \end{itemize}
\end{definition}

\begin{lemma}\label{lemma:unique_posterior}
    Given any $G$ and $\sigma$, every $H^{*}\in\argmin\nolimits_{H}\mathcal{L}^{\sigma}_{\textup{CEL}}(H,G)$ represents the posterior distribution $\mu_{\sigma}$ induced by $\sigma$ and $G$.
\end{lemma}

Lemma \ref{lemma:unique_posterior} shows that any BGP response $H^{*}$ to any $G$ represents the posterior belief induced by $G$ and $\sigma$.
In other words, the posterior belief is the unique functional minimizer of $\mathcal{L}^{\sigma}_{\texttt{CEL}}$.

\paragraph{Conditional Entropy}
In fact, the BNGP risk $\mathcal{L}^{\sigma}_{\texttt{CEL}}$ coincides with the \textit{conditional entropy}, denoted by $\mathcal{H}^{\sigma}_{G}(B;X)$, defined by
\[
\mathcal{H}^{\sigma}_{G}(B;X) = -\sum\nolimits_{b\in\mathcal{B}}\int\nolimits_{\mathcal{X}}\rho_{G}(x|b)\sigma(b)\log\left(h^{*}(b|x)\right)dx,
\]
where $\rho_{G}(x|b)$ is the underlying probability density function of the mechanism $\mathcal{M}(\cdot;G)$.

Theoretically, any $\sigma$-BNGP strategy $G$ with a given $\mathtt{UtilityL}$ can be derived by modeling the BGP response as the posterior distribution induced by $G$ and $\sigma$. However, this approach requires optimizing a conditional probability $\rho_G$, where both the objective function and constraints depend on its posterior distributions, governed by Bayes' rule. This introduces significant analytical and computational challenges. Unlike scenarios where $\rho_G(x|b)$ can be explicitly designed, such as in information design \cite{schmutte2022information}, many practical applications lack a closed-form expression for the underlying distribution of the mechanism $\mathcal{M}(\cdot;G)$ in terms of the privacy strategy, which complicates direct optimization.
It would become more challenging when we protect privacy for the composition of multiple heterogeneous mechanisms.

Directly involving the posterior distribution in the optimization process introduces significant computational challenges. In stochastic gradient methods, the gradient with respect to the privacy strategy $G$ depends on the posterior induced by $\rho_G$ via Bayes’ rule. When this posterior lacks a closed-form solution, each gradient update must approximate or simulate it—typically using methods like MCMC or variational inference—which are computationally expensive, particularly for complex or high-dimensional models. This repeated sampling increases computational costs, slows convergence, and introduces bias into the gradient estimates, potentially destabilizing training and leading the optimization toward suboptimal solutions.

Our BNGP approach circumvents the need for direct posterior modeling, thereby avoiding the computational and analytical challenges discussed above. By sidestepping explicit posterior estimation, our method can efficiently handle complex compositions of multiple heterogeneous mechanisms—even when these mechanisms interact. Moreover, since $\mathtt{M}(D, \cdot)$ can serve as a \textit{black-box data processing oracle}, our framework is capable of incorporating diverse real-world processing modules without requiring detailed knowledge of their internal operations. This flexibility not only simplifies the optimization process but also enhances the practicality and robustness of our privacy-preserving strategy.

\subsection{Properties of BGP Risk}\label{app:properties_of_BGP}

\textit{Post-processing} and \textit{composition} are two important properties for useful privacy frameworks such as the DP framework.
We show that the BGP risk (Definition \ref{def:BGP_response_risk}) satisfies the properties of post-processing and composition.

\subsubsection{Post-processing}
Let $\textup{Proc}: \mathcal{X}\mapsto \widetilde{\mathcal{X}}$ be any (possibly randomized) post-processing map. 
Given the original privatized output $x=\mathcal{M}(D;G)\in\mathcal{X}$, the post-processed output $\textup{Proc}(\mathcal{M}(D;G))$ lives in $\widetilde{\mathcal{X}}$.
In addition, let $\textup{Proc}\circ G$ denote the randomization that combines the randomness from $G$ and the post-processing.
That is, $\textup{Proc}\circ G$ takes membership $b$ and seed $r$ as input (via $G$), generates noise $\xi$, apply the $G$-perturbed $\mathcal{M}$, then post-processes via $\mathrm{Proc}$.

Proposition \ref{prop:post_processing} shows that the BGP risk satisfies the post-processing property.
That is, processing a data-sharing mechanism's output cannot increase input data information.

\begin{proposition}[Post-Processing]\label{prop:post_processing}
Suppose that $G$ has the BGP risk $H \in \arg\min\nolimits_{H''}\mathcal{L}^{\sigma}_{\texttt{CEL}}(H'', G)$.
Suppose in addition that for any $\textup{Proc}$, $\textup{Proc} \circ G$ has BGP risk $H' \in\arg\min\nolimits_{H''}\mathcal{L}_{\texttt{CEL}}(H'', \textup{Proc} \circ G)$. 
Then, $\mathcal{L}^{\sigma}_{\texttt{CEL}}(H', \textup{Proc} \circ G) \geq \mathcal{L}^{\sigma}_{\texttt{CEL}}(H,G)$.
\end{proposition}

\subsubsection{Composition}

Another key property of the BNGP risk is that it composes gracefully, ensuring privacy risk grows in a controlled, sublinear manner under sequential composition of mechanisms accessing the input dataset or related datasets.

Consider a profile $\vec{G} = \{G_{1}, \dots, G_{n}\}$ for $1 \leq n < \infty$, where each neural network $G_{j}$ represents the density function $g_{j}$. 
Let $\mathcal{M}_{j}(\cdot;G_{j}): \mathcal{D} \mapsto \mathcal{X}_{j}$ denote the $j$-th mechanism, for all $j \in [n]$, where $\mathcal{X}_{j}$ represents the output space of $\mathcal{M}_{j}$. Additionally, let $\rho_{j}: \mathcal{D} \mapsto \Delta(\mathcal{X}_{j})$ denote the underlying density function of $\mathcal{M}_{j}(\cdot;G_{j})$.

Define the composition $\mathcal{M}(\vec{G}): \mathcal{D} \mapsto \prod\nolimits_{j=1}^{n} \mathcal{X}_{j}$ of mechanisms $\mathcal{M}_{1}(\cdot;G_{1}), \dots, \mathcal{M}_{n}(\cdot;G_{n})$ as
\begin{equation}\label{eq:composition_def}
    \mathcal{M}(D; \vec{G}) \equiv \left(\mathcal{M}_{1}(D; G_{1}),  \dots, \mathcal{M}_{n}(D; G_{n})\right).
\end{equation}

The joint density function of $\mathcal{M}(\cdot; \vec{G})$, denoted by $\vec{\rho}: \mathcal{D} \mapsto \Delta\left(\prod\nolimits_{j=1}^{n} \mathcal{X}_{j}\right)$, encodes any underlying dependencies among the mechanisms. 
The mechanisms are independent if $\vec{\rho}(x_{1}, \dots, x_{n} | D) = \prod\nolimits_{j=1}^{n} \rho_{j}(x_{j} | D)$.
Otherwise, they are dependent. 
For simplicity, denote
\[
\vec{\mathtt{M}}(D, \vec{G}(b,\vec{r}))\equiv \left(\mathtt{M}_{1}\left(D, G_{1}(b, r_{1})\right), \dots, \mathtt{M}_{n}\left(D, G_{n}(b, r_{n})\right)\right).
\]
Let $\vec{x}=(x_{1}, \dots, x_{n}) = \vec{\mathtt{M}}(D, \vec{G}(b,\vec{r}))$.
Recall that each $\mathtt{M}_{j}(D, \xi_{i}) = \mathcal{M}_{i}(D; G_i)$ when noise $\xi_{i} = G_i(b, r_i)$ is generated, for $i\in[n]$.
Let $\vec{H}(\vec{x})$ denote the attacker's discriminator that uses all outputs (irrespective of their order), and let $H_{j}(x_{j})$ represent the discriminator that takes only $x_{j}$.

\begin{lemma}\label{lemma:composition}
    Suppose that $\mathcal{M}(\vec{G})$ is a composition of $n$ mechanisms with arbitrary correlation.
Then, for $\vec{H}^{*}\in \argmin\nolimits_{\vec{H}}\mathcal{L}^{\sigma}_{\texttt{CEL}}(\vec{H}, \vec{G})$, $H^{*}_{j}\in \argmin\nolimits_{H_{j}}\mathcal{L}^{\sigma}_{\texttt{CEL}}( H_{j},G_{j})$ for all $j\in[n]$,
\begin{align*}
    \mathcal{L}^{\sigma}_{\texttt{CEL}}(\vec{H}^{*},\vec{G})&= \sum^{n}\nolimits_{j=1}\mathcal{L}^{\sigma}_{\texttt{CEL}}(H^{*}_{j},G_{j}) + \Lambda^{\sigma}(\vec{G})\\
    &\leq \frac{1}{2}\mathbb{E}^{\vec{r}\sim \mathcal{U}^{n}}_{b\sim\sigma}\left[\log\Bigl((2\pi e)^K \det(\Sigma[\vec{G}\left(b,\vec{r}\right)])\Bigr)\right],
\end{align*}
where $K$ is the number of total individuals, $\Sigma[\vec{x}]$ is the conditional covariance matrix of the membership vector given the joint output $\vec{x}$, and
\[
\Lambda^{\sigma}(\vec{G}) =
\begin{cases}
    \mathcal{H}(Q), & \textup{ if mechanisms are \textup{independent}},\\
    \texttt{D}_{\texttt{KL}}\left(Q\| P\right), & \textup{ if mechanisms are \textup{dependent}}.
\end{cases}
\]
Here, $Q(\vec{x}) = \sum_{b} \vec{\rho}(\vec{x}|b)\sigma(b)$ and $P(\vec{x}) = \prod_{j=1}^n Q_j(x_j)$, where $Q_j(x_j) = \int_{\vec{\mathcal X}{-j}} Q(x_j, \vec{x}_{-j})\,d\vec{x}_{-j}$, with $\vec{\mathcal X}{-j}=\prod_{k\neq j}\mathcal{X}_{k}$ and $\vec{x}_{-j}$ is the joint outputs except $x_{j}$.
In addition, $\mathcal{H}(\cdot)$ is the differential entropy, and $\texttt{D}_{\texttt{KL}}(\cdot)$ is the Kullback–Leibler (KL) divergence.
\end{lemma}

\begin{theorem}[Additive Composition]\label{prop:composition_risk}
Suppose that $\mathcal{M}(\vec{G})$ is a composition of $n$ mechanisms with arbitrary correlation.
Let $\vec{H}^{*}\in \argmin\nolimits_{\vec{H}}\mathcal{L}^{\sigma}_{\texttt{CEL}}(\vec{H}, \vec{G})$ and $H^{*}_{j}\in \argmin\nolimits_{H_{j}}\mathcal{L}^{\sigma}_{\texttt{CEL}}( H_{j},G_{j})$ for all $j\in[n]$.
Suppose each $\mathcal{L}^{\sigma}_{\texttt{CEL}}(H^{*}_{j},G_{j}) = \textup{PB}_{j}$.
Then, $\mathcal{L}^{\sigma}_{\texttt{CEL}}(\vec{H}^{*},\vec{G}) = \sum^{n}_{j=1}\textup{PB}_{j} + \Lambda^{\sigma}(\vec{G})$.
\end{theorem}

Based on Lemma \ref{lemma:composition}, Theorem \ref{prop:composition_risk} shows that for any composition of $n$ mechanisms, the total BGP risk admits a additive aggregation: $\mathcal{L}^{\sigma}_{\texttt{CEL}}(\vec{H}^{*},\vec{G}) = \sum^{n}_{j=1}\textup{PB}_{j} + \Lambda^{\sigma}(\vec{G})$.
When mechanisms are independent, the term $\Lambda^{\sigma}(\vec{G})$ is the differential entropy of the joint output, which captures its intrinsic uncertainty.
In the presence of arbitrary correlations, $\Lambda^{\sigma}(\vec{G})$ equals $\texttt{D}_{\texttt{KL}}(Q\| P)$, the KL divergence between the true joint distribution and the product of the marginals, thereby quantifying precisely the extra information leaked by interdependencies.
Crucially, because $\Lambda^{\sigma}(\vec{G})$ measures residual uncertainty about a fixed-size membership vector $B$, it cannot grow proportionally with $n$; at worst it remains bounded by the secret's entropy, so the overall privacy risk increases only sublinearly in the number of composed mechanisms.

\subsection{Subjective Prior}\label{sec:subjective_prior}

The true membership prior $\theta\in \Delta(\{0,1\})^{K}$ specifies each individual's marginal probability of being selected to be included in the dataset.
The attacker's subjective prior $\sigma\in \Delta(\{0,1\})^{K}$ directly impacts MIA performance and thus the resulting privacy risk evaluated under the true prior $\theta$.
Proposition \ref{prop:static_subjective_prior} shows that, for any static (i.e., fixed) attack belief $\sigma$, using the BGP response matched to the true prior $\theta$ can never increase the privacy loss.

\begin{proposition}\label{prop:static_subjective_prior}
For any privacy strategy $G$, let $H^{*}_{\theta}\in\arg\min_{H}$ $\mathcal{L}^{\theta}_{\texttt{CEL}}(H, G)$, and let $H^{*}_{\sigma}\in\arg\min_{H}$ $\mathcal{L}^{\sigma}_{\texttt{CEL}}(H, G)$.
Then, for any static $\sigma\in \Delta(\{0,1\})^{K}$, we have $\mathcal{L}^{\theta}_{\texttt{CEL}}(H^{*}_{\theta}, G)\leq \mathcal{L}^{\theta}_{\texttt{CEL}}(H^{*}_{\sigma}, G)$.
\end{proposition}

In practice, an attacker may condition on side information $q$ about the true membership vector $b$ and thus hold priors more informative about the underlying $b$ than the true prior $\theta$.
Let $q$ be drawn from a signal space $\mathcal{V}$ according to a kernel $\beta(\cdot|b)$.
Given such a $q$, the attacker's subjective prior is therefore a mapping $\sigma\colon \mathcal V\longrightarrow \Delta(\{0,1\})^K$, $q\mapsto\bigl(\sigma_i(b_{i}\mid q)\bigr)_{i=1}^K$.
Let $\Sigma^{\beta,\sigma}\equiv\{\sigma(\cdot|q)\in \Delta(\{0,1\})^K, q\in \mathcal{V}, q\sim \beta\}$ be the full family of signal-indexed priors.
That is, the attacker's prior becomes a random variable induced by $\beta$ and $\theta$.
We then define the \textit{Bayes-plausible} subset 
$$\Sigma^{\beta,\sigma}_{\mathrm{BP}}\equiv\left\{
\begin{aligned}
    &\sigma(\cdot|q)\in\Sigma^{\beta,\sigma}:\\&\sum_{b'}\int_{\mathcal V}\sigma(b\mid q) \beta(dq|b')\theta(b')=\theta(b),\forall b\in\{0,1\}^K
\end{aligned}
\right\}.$$

When we consider such priors, rewrite $\mathcal{L}^{\sigma}_{\texttt{CEL}}$ in (\ref{eq:BCEL}) as:
\begin{equation*}
    \mathcal{L}^{\sigma,q}_{\texttt{CEL}}(H_{q},G)
        \equiv \mathbb{E}^{r\sim \mathcal{U}}_{b\sim \sigma(\cdot|q) }\left[\ell_{\texttt{CEL}}\left(H_{q}\left(\mathtt{M}\left(G(b, r) \right)\right), b\right)\right].
\end{equation*}
When the true membership vector is $b$, for any $H^{*}_{\sigma,q}\in\arg\min_{H}$ $\mathcal{L}^{\sigma,q}_{\texttt{CEL}}(H, G)$, define 
\begin{equation*}
    \mathcal{L}^{\theta}_{\texttt{CEL}}(H^{*}_{\sigma,q},G|b)
        \equiv \mathbb{E}^{r\sim \mathcal{U}}_{b'\sim \theta }\left[\ell_{\texttt{CEL}}\left(H^{*}_{\sigma,q}\left(\mathtt{M}\left(G(b', r) \right)\right), b'\right)\right].
\end{equation*}
Here, we have dependence on $b$ because $q$ that is drawn by $\beta(\cdot|b)$ implicitly depends on $b$.

\begin{proposition}\label{prop:Bayes_plausibility}
Let $\sigma$ be an attacker's subjective prior.
    For any privacy strategy $G$, let $H^{*}_{\theta}\in\arg\min_{H}$ $\mathcal{L}^{\theta}_{\texttt{CEL}}(H, G)$, and let $H^{*}_{\sigma,q}\in\arg\min_{H}$ $\mathcal{L}^{\sigma,q}_{\texttt{CEL}}(H, G)$.
    Then, $\mathcal{L}^{\theta}_{\texttt{CEL}}(H^{*}_{\theta}, G)$ $\geq \mathcal{L}^{\theta}_{\texttt{CEL}}(H^{*}_{\sigma,q}, G|b)$ for all $q\in\mathcal{V}$, $b\in W$ if and only if for $\beta$-almost every $q\in\mathcal{V}$, $\sigma(\cdot|q)\in\Sigma_{\mathrm{BP}}$.
\end{proposition}

Proposition \ref{prop:static_subjective_prior} shows that if an attacker fixes any prior $\sigma\neq \theta$, the defender's $\theta$-BGP response always incurs no greater expected privacy loss (in terms of $\mathcal{L}^{\theta}_{\texttt{CEL}}$) under the true prior $\theta$ than the response tuned to $\sigma$.
Proposition \ref{prop:Bayes_plausibility} sharpens this to the per-state level: an attacker who adapts their prior to a signal $q$ can reduce the defender’s pointwise loss only if that signal-conditioned posterior is a Bayes-plausible refinement of $\theta$.  Any \textit{adaptive} prior that fails the Bayes-plausibility test cannot improve inference over the true-prior predictor.

\subsection{Mechanism Comparison in BNGP privacy}

Comparing privacy mechanisms on both their formal guarantees and empirical utility is inherently difficult.
There is no single metric or ranking that holds across all settings.
For example, differential privacy (DP) typically uses a parameter pair $(\epsilon, \delta)$ to quantify privacy loss \cite{abadi2016deep,papernot2021tempered}.
Yet two mechanisms with the same $(\epsilon, \delta)$ can differ drastically in how that budget is spent across individual queries or data points \cite{kaissis2024beyond}.
Likewise, even fixing $(\epsilon,\delta)$, different noise distributions yield different utility-privacy trade-offs for the same data processing applications \cite{geng2014optimal}.

Let $\mathcal{G}$ represent the set of privacy strategies we wish to compare.
To make any two strategies in $\mathcal{G}$ meaningfully comparable, we impose the following regularity condition in Assumption \ref{assp:regular_condition} on utility loss $\mathtt{UtilityL}(\cdot)$. 
This assumption does not affect the derivation of the BNGP strategy for any defender-specific definition of utility loss.

\begin{assumption}\label{assp:regular_condition}
For fixed attacker strategy $H$ and prior distribution $\sigma$, if privacy strategies $G$ and $G'$ satisfy $\mathcal{L}^{\sigma}_{\texttt{CEL}}(H,G) \geq \mathcal{L}^{\sigma}_{\texttt{CEL}}(H,G')$, then it also satisfies $\mathtt{UtilityL}(G) \geq \mathtt{UtilityL}(G')$. Additionally, strategies with equal utility loss must have identical CEL loss.
\end{assumption}

This assumption simply enforces a privacy–utility monotonicity on our strategy class: whenever one mechanism is strictly less private (i.e., has higher conditional‐entropy loss), it cannot be strictly more accurate (i.e., have lower utility loss), and conversely, equal utility implies equal privacy loss.  Standard metrics, such as mean squared error, absolute error, classification accuracy, and negative log-likelihood, satisfy this assumption.

\begin{proposition}\label{thm:BNGP}
Suppose Assumption \ref{assp:regular_condition} holds.
Let $G^{*}\in\mathcal{G}$ be a $\sigma$-BNGP optimal strategy, and recall $\widehat{\mathcal{L}}^{\sigma}_{\texttt{def}}(H, G | Z)$ from (\ref{eq:trade_off_PB}) and (\ref{eq:trade_off_UB}). Let $H_{G}$ represent the posterior distribution induced by strategy $G$ and prior $\sigma$. For any alternative strategy $G'$ that minimizes the defender’s objective function $\widehat{\mathcal{L}}^{\sigma}_{\texttt{def}}$ with attacker’s optimal response $H'$, we have:
\[
\mathcal{L}^{\sigma}_{\texttt{CEL}}(H_{G^{*}},G^{*}) \leq \mathcal{L}^{\sigma}_{\texttt{CEL}}(H_{G'},G').
\]
\end{proposition}

Proposition \ref{thm:BNGP} relies on Lemma \ref{lemma:unique_posterior}, which establishes a foundational property of the Bayes Generalized Privacy (BGP) response. Specifically, that the BGP risk equals conditional entropy. Thus, Proposition \ref{thm:BNGP} ensures that among all strategies optimizing the defender's objective, the BNGP strategy $G^{*}$ minimizes the worst-case privacy loss, independent of the attacker’s preference parameter $\gamma$ weighting TPR and FPR (see, (\ref{eq:BW_MA})).

\subsubsection{Comparison of DP Mechanisms}

Define $\mathcal{A}^{k}(d_k, x) \in \{0,1\}$ as the attacker’s binary decision for individual $k$. Introduce the maximum $(\sigma,\gamma)$-Bayesian Worst-Case Membership Advantage (BWMA):
\[
\overline{\mathtt{Adv}}^{\gamma}_{\sigma}(G) \equiv \max\nolimits_{H}\mathtt{Adv}^{\gamma}_{\sigma}(H,G).
\]
For given $\epsilon \geq 0$, let $\delta(\epsilon)$ denote the privacy profile \cite{balle2020privacy}, i.e., the smallest $\delta$ such that $\mathcal{M}(\cdot; G)$ satisfies $(\epsilon, \delta)$-DP.

\begin{theorem}\label{thm:ordering_DP_BGP}
Consider mechanisms $\mathcal{M}(\cdot;G)$ and $\mathcal{M}(\cdot;G')$, satisfying $(\epsilon, \delta(\epsilon))$-DP and $(\epsilon', \delta'(\epsilon'))$-DP, respectively. Denote their BGP risks as $\mathcal{L}^{\sigma}_{\texttt{CEL}}(G)$ and $\mathcal{L}^{\sigma}_{\texttt{CEL}}(G')$. The following statements are equivalent:
\begin{itemize}
\item[(1)] For all prior distributions $\sigma \in \Theta$, $\mathcal{L}^{\sigma}_{\texttt{CEL}}(G) \geq \mathcal{L}^{\sigma}_{\texttt{CEL}}(G')$.
\item[(2)] For all $\epsilon \geq 0$, the privacy profile of $G$ satisfies $\delta(\epsilon) \geq \delta'(\epsilon)$.
\end{itemize}
Moreover, either (1) or (2) implies:
\begin{itemize}
\item[(3)] For all $\sigma \in \Theta$ and $\gamma \in (0,1]$, $\overline{\mathtt{Adv}}^{\gamma}_{\sigma}(G) \leq \overline{\mathtt{Adv}}^{\gamma}_{\sigma}(G')$.
\end{itemize}
\end{theorem}

Theorem \ref{thm:ordering_DP_BGP} establishes a clear relationship among different privacy metrics by showing that two fundamental conditions are equivalent, and that either of them guarantees a desirable property in terms of membership inference. Specifically, condition (1) states that for every prior distribution $ \sigma \in \Theta $, the BGP risk of $\mathcal{M}(\cdot;G)$ is uniformly higher than that of $\mathcal{M}(\cdot;G')$, which implies that $\mathcal{M}(\cdot;G)$ leads to a lower worst-case privacy loss. Condition (2) provides an equivalent characterization in terms of the privacy profile: for every $\epsilon \geq 0$, the privacy leakage $\delta(\epsilon)$ associated with $\mathcal{M}(\cdot;G)$ is no smaller than that of $\mathcal{M}(\cdot;G')$, ensuring uniformly better privacy guarantees across all values of $\epsilon$. 

Moreover, either of these conditions implies condition (3), which asserts that for all $\sigma \in \Theta$ and for any preference parameter $\gamma \in (0,1]$, the maximum membership inference advantage $\overline{\mathtt{Adv}}^{\gamma}_{\sigma}(G)$ is no larger than $\overline{\mathtt{Adv}}^{\gamma}_{\sigma}(G')$. In other words, if a mechanism exhibits stronger privacy according to either the BGP risk or the privacy profile, it will also result in a lower maximum advantage for an attacker performing membership inference. However, it is important to note that while conditions (1) and (2) are equivalent, condition (3) is only a one-way implication and does not necessarily imply the other two.

\subsection{Approximating the BGP Risk}\label{sec:approximating}

The optimization problem defining the Bayesian Nonparametric Gaussian Process (BNGP) strategy is inherently infinite-dimensional and, therefore, not directly solvable using standard computational methods. To address this issue, we approximate the BNGP strategy using neural networks with finite parameterization. Specifically, let $G_{\lambda_{\mathtt{d}}}(b,r)$ represent a neural network approximation of the privacy strategy $G(b,r)$ with parameters $\lambda_{\mathtt{d}}$, and let $H_{\lambda_{\mathtt{a}}}(x)$ denote a neural network approximation of the attacker’s strategy $H(x)$, parameterized by $\lambda_{\mathtt{a}}$.

Approximating the attacker’s strategy $H$ with a finite-parameter neural network introduces approximation errors, potentially causing discrepancies between the computed and true worst-case privacy risks. Because the BNGP strategy inherently depends on accurately modeling the attacker’s optimal response, the effectiveness and robustness of the approximated defender strategy $G_{\lambda_{\mathtt{d}}}$ directly depend on the quality and precision of the neural network approximation $H_{\lambda_{\mathtt{a}}}$. Therefore, we investigate the consequences of finite parameterization, explicitly quantifying the resulting approximation errors in $H_{\lambda_{\mathtt{a}}}$ and their implications for privacy risks.

\paragraph{Neural Network Structure}
We use a feedforward neural network $H_{\lambda_{\mathtt{a}}}: \mathcal{X}\mapsto [0,1]^{K}$ with $\mathtt{l}$ hidden layers of width $\mathtt{w}$, which outputs a vector of individual membership probabilities $(\Pr[b_k = 1|x])^K_{k=1}$.

\paragraph{Activation Functions}
Each hidden layer employs the Rectified Linear Unit (ReLU) activation function, chosen for its computational simplicity and strong approximation capabilities. The final output layer uses a sigmoid activation function, applied coordinate-wise, ensuring the outputs are valid probability estimates bounded within the interval $[0,1]$.

\paragraph{Effective Intrinsic Dimension}
Let $\mu_{\lambda_{\mathtt{d}}}:\mathcal{X}\mapsto \Delta(\{0,1\})^{K}$ denote the posterior distribution induced by the neural network representation $G_{\lambda_{\mathtt{d}}}$ and the prior distribution $\sigma$. We define the \textit{effective intrinsic dimension}, denoted $\alpha_{\lambda_{\mathtt{d}}}$, as the smallest integer $\alpha\leq K$ such that for every $\eta>0$ there exist
\begin{itemize}
    \item[(i)] a mapping $F:\mathcal{X}\mapsto \mathbb{R}^{\alpha}$,
    \item[(ii)] a Lipschitz-continuous mapping $T: \mathbb{R}^{\alpha} \mapsto \Delta(\{0,1\})^{K}$,
\end{itemize}
that satisfies $\sup_{x\in\mathcal{X}} \mathcal{W}_{1}\left(\mu_{\lambda_{\mathtt{d}}}(x), T(F(x))\right)\leq\eta$, where $\mathcal{W}_{1}$ denotes the Wasserstein-$1$ distance. This effective intrinsic dimension captures the minimal latent dimension needed to approximate the posterior within $\eta$ accuracy across all $x$.

\begin{proposition}\label{prop:error_bound_H}
Suppose the neural network width satisfies $\mathtt{w}\geq 7K+1$ and the posterior distribution $\mu_{\lambda_{\mathtt{d}}}$ is uniformly Lipschitz continuous. Let $h_{\lambda_{\mathtt{a}}}$ represent the approximate distribution given by the neural network $H_{\lambda_{\mathtt{a}}}$. Then, for any fixed $G_{\lambda_{\mathtt{d}}}$, the Wasserstein-$1$ distance between the true posterior distribution $\mu_{\lambda_{\mathtt{d}}}$ and its approximation $h_{\lambda_{\mathtt{a}}}$ is bounded as:
\[
\mathcal{W}_{1}(\mu_{\lambda_{\mathtt{d}}}(x), h_{\lambda_{\mathtt{a}}}(x))\leq 0.25\mathtt{C}\left(\mathtt{w}^{2}\mathtt{l}\right)^{-1/\alpha_{\lambda_{\mathtt{d}}}}, \quad \forall x\in\mathcal{X},
\]
where the constant $\mathtt{C}$ depends solely on intrinsic properties of $\mu_{\lambda_{\mathtt{d}}}$ (e.g.\ its covering number in the $\alpha_{\lambda_{\mathtt{d}}}$-dimensional latent space), including its uniform Lipschitz constant and the geometric characteristics of its low-dimensional representation provided by $F$. Importantly, $\mathtt{C}$ is independent of neural network parameters such as width $\mathtt{w}$ and depth $\mathtt{l}$.
\end{proposition}

Proposition \ref{prop:error_bound_H} shows that the attacker-posterior approximation error obeys a power-law decay in neural network capacity: $\mathcal{W}_{1}(\mu_{\lambda_{\mathtt{d}}}(x), h_{\lambda_{\mathtt{a}}}(x)) = O((\mathtt{w}^{2})^{-1/\alpha_{\lambda_{\mathtt{d}}}})$.
Here, $\alpha_{\lambda_{\mathtt{d}}}$ is the effective intrinsic dimension, so the exponent $-1/\alpha_{\lambda_{\mathtt{d}}}$ governs the rate at which error shrinks as you increase width $\mathtt{w}$ or depth $\ell$.
Because the constant $\mathtt{C}$ depends only on the intrinsic properties of $\mu_{\lambda_{\mathtt{d}}}$, this bound guarantees that any increase in network size yields a predictable reduction in approximation error.
In turn, more accurate modeling of the attacker's optimal response produces more faithful estimates of the true privacy loss; so larger networks directly translate into stronger, more reliable privacy guarantees for the defender's strategy.

\section{Experiments}\label{sec:experiments}

\subsection{Example: MIA in Sharing Summary Statistics}\label{sec:case_study}

We apply the BNGP framework to the sharing of summary statistics from binary datasets, as outlined in the example in Section \ref{sec:MIA}.
Assuming the attributes in each $d_{k}$ are independent.
SNVs can be prefiltered to retain only those in \textit{linkage equilibrium} \cite{kimura1965attainment}.
An MIA attacker uses the summary statistics $x$ output by $f(D)$ to infer whether specific individuals $k \in \mathcal{K}$ belong to the private dataset $D$. 
We compare our Bayesian model with state-of-the-art (SOTA) Frequentist attacks, including \textit{fixed(-threshold) LRT} \cite{sankararaman2009genomic, shringarpure2015privacy, venkatesaramani2021defending, venkatesaramani2023enabling}, \textit{adaptive LRT} \cite{venkatesaramani2021defending, venkatesaramani2023enabling}, and the \textit{optimal} LRT. 
These attacks rely on the log-\textit{likelihood ratio statistic} $\mathtt{lrs}(d_k, x)$, which compares observed summary statistics $x$ to \textit{reference frequencies} $\bar{p}_j$ derived from a population dataset independent of the population dataset $Z$. Detailed definitions of these models and loss functions are provided in Appendix \ref{app:existence_lrt_attack}.

The fixed LRT attacker determines whether individual $k$ is part of the dataset by rejecting $H_0^k: b_{k}=0$ (absence) in favor of $H_1^k:b_{k}=1$ (presence) if $\mathtt{lrs}(d_k, x) \leq \tau$, where the fixed $\tau$ balances Type-I ($\alpha_{\tau}$) and Type-II ($\beta_{\tau}$) errors. The adaptive LRT dynamically adjusts $\tau^{(N)}$ using reference population data to refine the hypothesis test. The optimal LRT minimizes Type-II error $\beta_{\tau^{*}}$ for a given $\alpha_{\tau^{*}}$, achieving the most powerful test by Neyman-Pearson lemma \cite{neyman1933ix}. \textit{Optimal $\alpha$-LRT attacks} refer to Likelihood Ratio Tests that are Neyman-Pearson optimal at a fixed significance level $\alpha$.

\subsection{Experimental Results}

\begin{figure*}[htb]
    \centering
    \begin{subfigure}[b]{0.3\textwidth}
        \includegraphics[width=\textwidth]{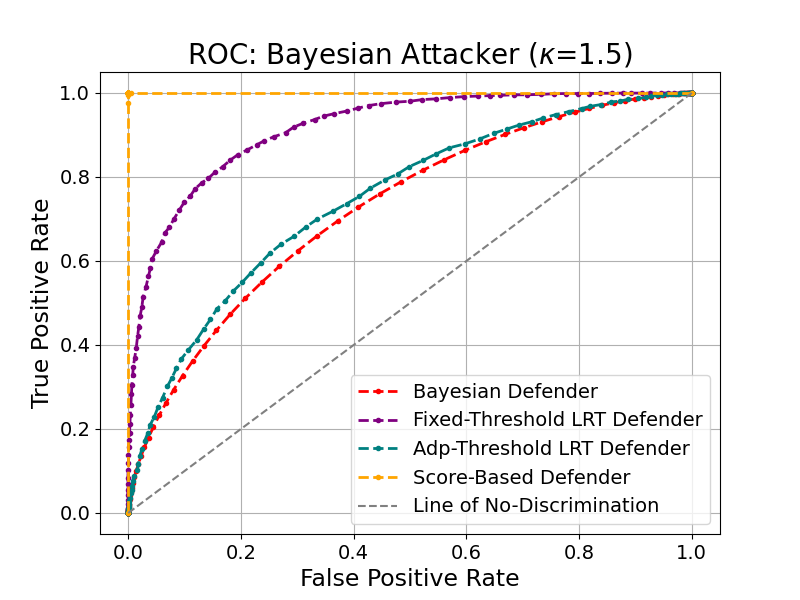}
        \caption{}
        \label{fig:fig4}
    \end{subfigure}
    \hfill
    \begin{subfigure}[b]{0.3\textwidth}
        \includegraphics[width=\textwidth]{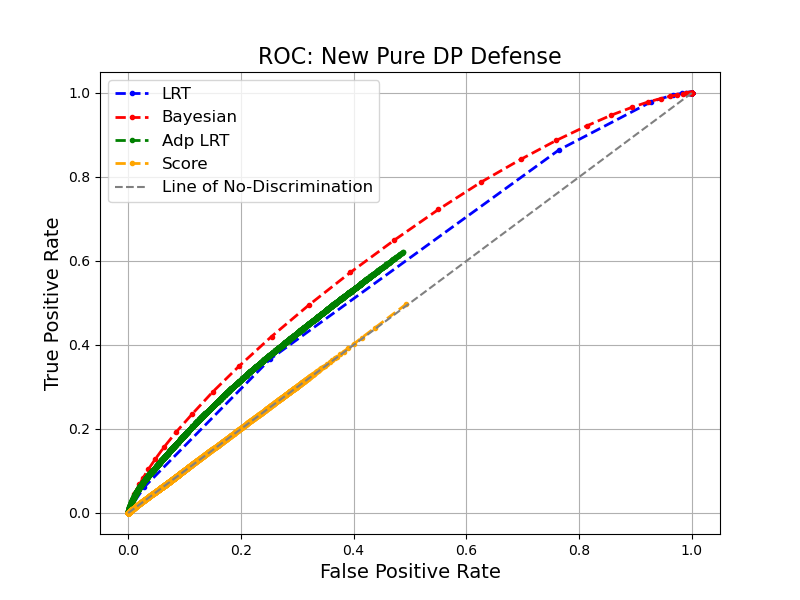}
        \caption{}
        \label{fig:fig2x}
    \end{subfigure}
    \hfill
    \begin{subfigure}[b]{0.3\textwidth}
        \includegraphics[width=\textwidth]{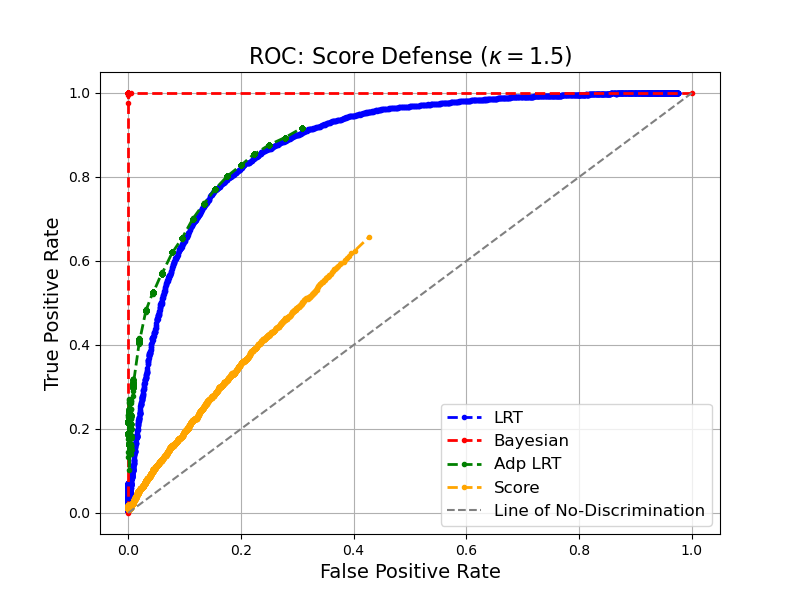}
        \caption{}
        \label{fig:fig4x}
    \end{subfigure}
    
    \vspace{0.5em}

    \begin{subfigure}[b]{0.3\textwidth}
        \includegraphics[width=\textwidth]{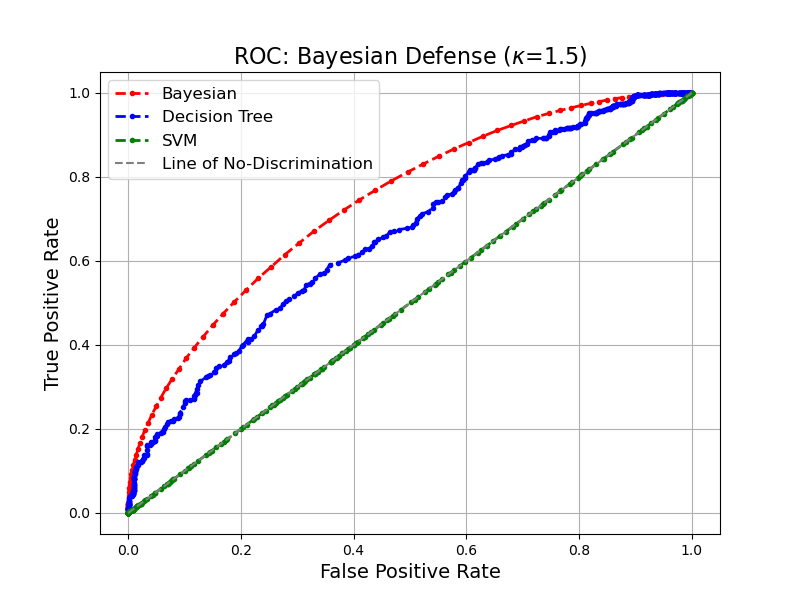}
        \caption{}
        \label{fig:fig5x}
    \end{subfigure}
    \hfill
    \begin{subfigure}[b]{0.3\textwidth}
        \includegraphics[width=\textwidth]{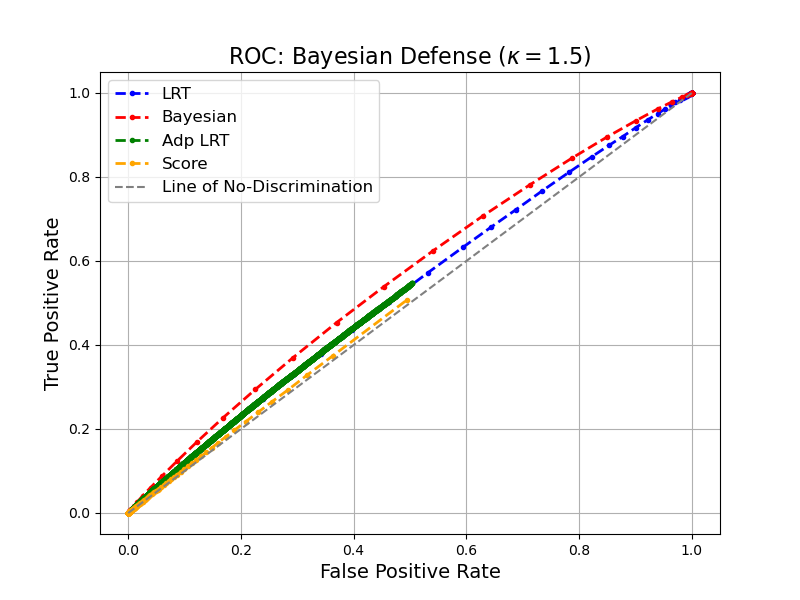}
        \caption{}
        \label{fig:fig1x}
    \end{subfigure}
    \hfill
    \begin{subfigure}[b]{0.3\textwidth}
        \includegraphics[width=\textwidth]{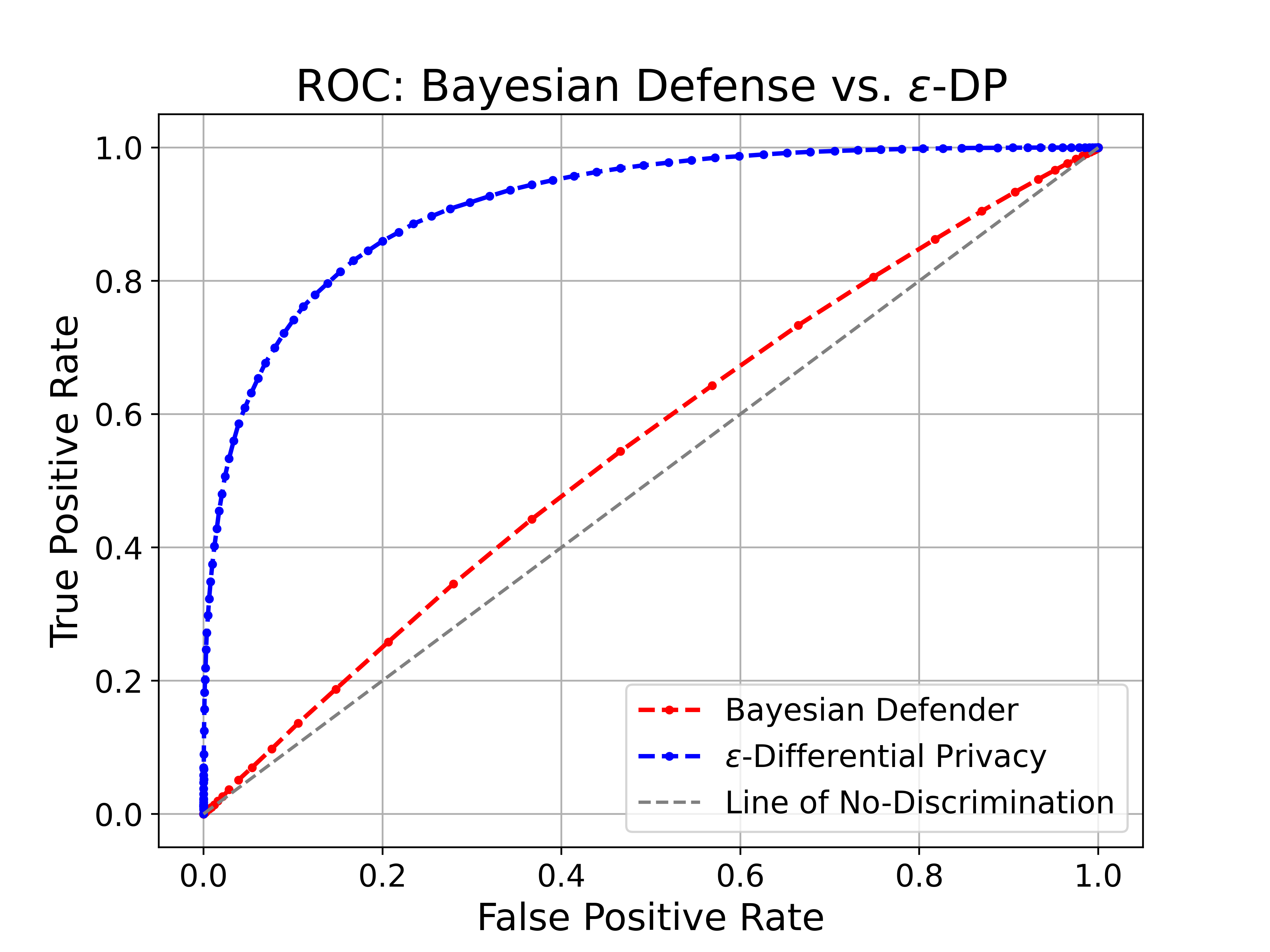}
        \caption{}
        \label{fig:fig6}
    \end{subfigure}
    
    \caption{
    (a) Performance of a Bayesian attacker against various defenders on a genomic dataset with 5000 SNVs per individual.\
    (b) Performance of a new Pure DP defender against different attackers on a genomic dataset with 100 SNVs.\
    (c) Performance of a score-based defender against various attackers using a genomic dataset with 4000 SNVs from 5 individuals (note: all other genomic data experiments used 800 individuals).\
    (d) Performance of a Bayesian defender against different attackers on the MNIST dataset.\
    (e) Performance of a Bayesian defender against different attackers on the Adult dataset.\
    (f) Performance of a Bayesian attacker under the BNGP defender and the DP defender with non-homogeneous SNP preferences (i.e., the defender’s parameter $\kappa$ assigns different weights to different attributes) when both defenders incur the same utility loss on a genomic dataset with 5000 SNVs.
    }
    \label{fig:six_figures}
\end{figure*}

Our experiments use three datasets: the \textit{Adult dataset} (UCI Machine Learning Repository), the \textit{MNIST dataset}, and a \textit{genomic dataset}. Detailed experimental setups and additional results are provided in Appendix \ref{sec:app_Hyperparameters}. 
We refer to an attacker leveraging BGP responses as the \textit{Bayesian attacker} and a defender employing a BNGP strategy as the \textit{Bayesian defender}.
In the experiments, we consider that the defender's loss function $\mathcal{L}_{\texttt{def}}$ takes the form of (\ref{eq:additive_obj}).

We consider the following baseline attack models: fixed-threshold LRT (LRT) and adaptive-threshold LRT (Adp LRT) attackers \cite{sankararaman2009genomic, shringarpure2015privacy, venkatesaramani2021defending, venkatesaramani2023enabling}, the score-based attacker \cite{dwork2015robust}, and decision-tree and support vector machine (SVM) attackers.
The score-based attacker, introduced by \cite{dwork2015robust}, generalizes the LRT attacker from \cite{homer2008resolving} by requiring only that the distorted summary statistics approximate the true marginals in the $L_1$-norm.
We refer to a defender who designs a privacy strategy against an $\mathtt{X}$ attacker as the $\mathtt{X}$ defender, where $\mathtt{X}$ represents one of the baseline attackers described above.
In addition to the $\mathtt{X}$ defender, we also consider the baseline defenders using \textit{standard DP} and \textit{new pure DP} \cite{steinke2016between}.

Figure \ref{fig:fig4} illustrates the Bayesian attacker’s performance across three scenarios, each with the mechanism protected by a different defender, using the genomic dataset.
The Bayesian defender employs the BNGP strategy, while the fixed-threshold LRT and adaptive LRT defenders adopt privacy strategies that best respond to their respective LRT attackers.
The results demonstrate that the BNGP strategy is the most robust defense against the Bayesian attacker among the three defense models.

Figure \ref{fig:fig2x} illustrates the performance of four attackers—Bayesian, fixed-threshold LRT, adaptive LRT, and score-based—each representing a different privacy risk scenario for the defenders on a genomic dataset. The defender employs the new pure DP defense strategy \cite{steinke2016between}. The results indicate that the Bayesian attacker, which uses a BGP response, achieves the highest performance, followed by the adaptive LRT (which slightly outperforms the fixed-threshold LRT), while the score-based attacker performs the worst. These findings demonstrate that the Bayesian attacker poses the worst-case privacy risk under the same new pure DP defense strategy among the evaluated baselines.

Figure \ref{fig:fig4x} compares the performance of four attackers—Bayesian, fixed-threshold LRT, adaptive LRT, and score-based—when the mechanism is defended by a score-based strategy (i.e., one that best responds to the score-based attacker) using the genomic dataset. The results show that the Bayesian attacker, employing BGP responses, significantly outperforms the others. Both the fixed-threshold and adaptive LRT attackers perform similarly but are notably less effective, while the score-based attacker achieves the poorest performance under the same defender. These findings echo the observations in Figure \ref{fig:fig2x} for the score-based defender.

Figure \ref{fig:fig5x} shows the MNIST classifier’s performance under privacy attacks. The ROC curve compares three attackers—Bayesian, decision-tree, and SVM—against a mechanism protected by a Bayesian defender using the BNGP strategy. The Bayesian attacker outperforms the others, followed by the decision-tree, while the SVM attacker performs worst. These results suggest that attackers using advanced machine-learning techniques do not represent the worst-case scenario; thus, a defender considering only non-Bayesian attackers may achieve less robust privacy protection.

Figure \ref{fig:fig1x} further demonstrates the robustness of the BGP response by presenting the performance of four attackers on the Adult dataset. In this experiment, the mechanism is defended by a Bayesian defender employing the BNGP strategy. The results indicate that the Bayesian attacker outperforms the others, with the adaptive LRT slightly surpassing the fixed-threshold LRT, and the score-based attacker yielding the weakest performance.

Figure \ref{fig:fig6} compares the performance of a Bayesian attacker under two defenses applied to a genomic dataset. In one case, the Bayesian defender employs the BNGP strategy; in the other, a conventional $\epsilon$-DP mechanism is used. Detailed information regarding the experimental setup is provided in Appendix \ref{sec:app_naive_DP}.
The Bayesian defender incorporates heterogeneous privacy-utility trade-offs by assigning weights $\vec{\kappa} = (\kappa_{j})_{j \in Q}$ to SNV positions. Specifically, $\kappa_j = 0$ for $90\%$ of the 5000 SNVs and $\kappa_j = 50$ for the remaining $10\%$, indicating that utility loss is a concern only for this minority subset. In contrast, the conventional $\epsilon$-DP strategy does not consider such preferences; instead, its privacy parameter $\epsilon$ is chosen to equate the expected utility loss to that of the Bayesian defender.
In our experiments, the BNGP strategy incurs a utility loss of approximately $0.0001$, with the corresponding $\epsilon$ set to $1.25 \times 10^{5}$ (see Appendix \ref{sec:app_naive_DP} for further details). The results demonstrate that, despite equivalent utility loss, the $\epsilon$-DP defense suffers a substantially higher privacy loss under the Bayesian attack—evidenced by an AUC of $0.53$ for the Bayesian defender compared to $0.91$ for the $\epsilon$-DP defender.

\section{Conclusion}\label{sec:conclusion}

This paper presents a game-theoretic framework for optimal privacy-utility trade-offs that overcomes differential privacy’s limitations. By modeling privacy protection as a Bayesian game, we derive the Bayes-Nash Generative Privacy (BNGP) strategy, which tailors trade-offs to defender preferences without intractable sensitivity calculations, supports complex compositions, and is robust to heterogeneous attacker preferences. Empirical results confirm BNGP’s effectiveness in privacy-preserving data sharing and classification, establishing it as a flexible, practical alternative to existing methods.

\bibliographystyle{IEEEtran}
\bibliography{references}

\newpage

\appendices

\section{Proof of Proposition \ref{prop:weighted_membership_advantage}}\label{app:prop:weighted_membership_advantage}

Let
\[
v(s,b) = \sum_{k\in \mathcal{K}} s_k\,b_k \quad \text{and} \quad c=\sum_{k\in \mathcal{K}} s_k,
\]
so that the attacker's loss is defined as
\[
\begin{aligned}
    \ell^{\gamma}_{\texttt{att}}(s,b) &\equiv -v(s,b) + \gamma\,c(s) = -\sum_{k\in \mathcal{K}}\Bigl(s_k b_k - \gamma s_k\Bigr) \\
    &= -\sum_{k\in \mathcal{K}} s_k (b_k-\gamma).
\end{aligned}
\]
Since \(s_k\) is a binary indicator, we have 
\[
s_k(b_k-\gamma) = \mathbf{1}_{\{s_k=1\}} \Bigl[(b_k-\gamma)\Bigr].
\]
Note that when \(b_k=1\), \(b_k-\gamma = 1-\gamma\); when \(b_k=0\), \(b_k-\gamma = -\gamma\). Thus, we can write
\[
s_k(b_k-\gamma) = (1-\gamma)\,\mathbf{1}_{\{s_k=1\}}\mathbf{1}_{\{b_k=1\}} - \gamma\,\mathbf{1}_{\{s_k=1\}}\mathbf{1}_{\{b_k=0\}}.
\]
Therefore,
\[
\begin{aligned}
    &\ell^{\gamma}_{\texttt{att}}(s,b) \\
    &= -\sum_{k\in \mathcal{K}} \Bigl[ (1-\gamma)\,\mathbf{1}_{\{s_k=1\}}\mathbf{1}_{\{b_k=1\}} - \gamma\,\mathbf{1}_{\{s_k=1\}}\mathbf{1}_{\{b_k=0\}} \Bigr].
\end{aligned}
\]
The expected loss is defined as
\[
\mathcal{L}^{\gamma,\sigma}_{\texttt{att}}(h, g) \;=\; \sum_{s,b}\int_{x} \ell^{\gamma}_{\texttt{att}}(s,b)(s,b) \, h(s|x)\, \rho_{g}(x|b) \, dx\, \sigma(b).
\]
Substituting the expression for $\ell^{\gamma}_{\texttt{att}}(s,b)$, we obtain
\[
\begin{aligned}
    \ell^{\gamma}_{\texttt{att}}(s,b) &= -\sum_{s,b}\int_{x} \sum_{k\in\mathcal{K} } \Big[\\
    &(1-\gamma)\,\mathbf{1}_{\{s_k=1\}}\mathbf{1}_{\{b_k=1\}} - \gamma\,\mathbf{1}_{\{s_k=1\}}\mathbf{1}_{\{b_k=0\}} \Big]\\
    &\times h(s|x) \,\rho_{g}(x|b) \, dx\, \sigma(b).
\end{aligned}
\]
For each individual \(k\), split the summation over \(s\) and \(b\) as follows. Denote by \(s_{-k}\) and \(b_{-k}\) the vectors for the other individuals. Then the above sum can be written as a sum over \(k\) of terms of the form
\[
\begin{aligned}
    &-\sum_{s_k, b_k}\left[ (1-\gamma)\,\mathbf{1}_{\{s_k=1\}}\mathbf{1}_{\{b_k=1\}} - \gamma\,\mathbf{1}_{\{s_k=1\}}\mathbf{1}_{\{b_k=0\}} \right] \\
    &\times\left( \sum_{s_{-k},b_{-k}} h(s_k,s_{-k}|x) \,\rho_{g}(x|b) \,\sigma(b) \right).
\end{aligned}
\]
Define the TPR and FPR as
\[
\mathtt{TPR}(h, g) \equiv \sum_{b_{-k}} \Pr\left[s_k=1 \mid b_k=1, x\right]\,\sigma(b_k=1,b_{-k}),
\]
and
\[
\mathtt{FPR}(h, g) \equiv \sum_{b_{-k}} \Pr\left[s_k=1 \mid b_k=0, x\right]\,\sigma(b_k=0,b_{-k}).
\]
Thus, taking the expectation over $x$ with respect to $\rho_{g}(x|b)$ and summing over all \(k\) yields
\[
\mathcal{L}^{\gamma,\sigma}_{\texttt{att}}(h, g) = -\Bigl[(1-\gamma)\,\mathtt{TPR}(h, g) - \gamma\,\mathtt{FPR}(h, g)\Bigr].
\]
By definition, the Bayes-weighted membership advantage is given by
\[
\mathtt{Adv}^{\gamma}_{\sigma}(h, g) \equiv (1-\gamma)\,\mathtt{TPR}(h, g) - \gamma\,\mathtt{FPR}(h, g).
\]
Hence, we have $\mathcal{L}^{\gamma,\sigma}_{\texttt{att}}(h, g) = -\mathtt{Adv}^{\gamma}_{\sigma}(h_A, g_D)$.

\qed

\section{Proof of Lemma \ref{lemma:unique_posterior}}

To prove Lemma \ref{lemma:unique_posterior}, we directly work on the functions of the probabilities $h$ and $g$ represented by $H$ and $G$, respectively, where $h:\mathcal{X}\mapsto [0,1]^{K}$.
Let $h(x)=q(x)=(q_{1}(x), \dots, q_{K}(x))$, where each $q_{k}(x)\in[0,1]$ is the probability of $b_{k}=1$. 
With abuse of notation, we let $h(b|x)=(h(b_{1}|x), \dots, h(b_{K}|x))$, where $h(b_{k}=1|x)=q_{k}(x)$ and $h(b_{k}=0|x)=1-q_{k}(x)$.
Then, the expected CEL can be written as
\[
\begin{aligned}
    \mathcal{L}^{\sigma}_{\texttt{CEL}}\left(h,g\right)= -\sum_{b}\int_{x}\sigma(b)\log\left(h(b|x)\right)\rho_{g}(x|b)dx.
\end{aligned}
\]

Let $\mu^{\sigma}:\mathcal{X}\mapsto \Delta(W)$ denote the posterior distribution induced by $\sigma$ and $g$ according to the Bayes' rule.
Then,
\[
\begin{aligned}
    &\mathcal{L}^{\sigma}_{\texttt{CEL}}\left(h,g\right) - \mathcal{L}^{\sigma}_{\texttt{CEL}}\left(\mu^{\sigma},g\right)\\
    &= - \sum_{b}\int_{x}\sigma(b)\log\left(h(b|x)\right)\rho_{g}(x|b)dx \\
    &+ \sum_{b}\int_{x}\sigma(b)\log\left(\mu^{\sigma}(b|x)\right)\rho_{g}(x|b)dx\\
    &=\sum_{b}\int_{x}\sigma(b)\rho_{g}(x|b)\left[\log\left(\mu^{\sigma}(b|x)\right) - \log\left(h(b|x)\right) \right]dx\\
    &=\sum_{b}\int_{x}\sigma(b)\rho_{g}(x|b)\log\left(\frac{\mu^{\sigma}(b|x)}{h(b|x)}\right)dx.
\end{aligned}
\]
By Bayes' rule, we have 
\[
\sigma(b)\rho_{g}(x|b) = \mu^{\sigma}(b|x)P^{\sigma}(x), 
\]
where $P^{\sigma}(x)=\sum_{b}\sigma(b)\rho_{g}(x|b)$. 
Then, we have
\[
\begin{aligned}
&\mathcal{L}^{\sigma}_{\texttt{CEL}}\left(h,g\right) - \mathcal{L}^{\sigma}_{\texttt{CEL}}\left(\mu^{\sigma},g\right) \\&= \sum_{b}\int_{x}\mu^{\sigma}(x)\log\left(\frac{\mu^{\sigma}(b|x)}{h(b|x)}\right)P^{\sigma}(x)dx\\
&\geq 0,
\end{aligned}
\]
which is non-negative because it is the Kullback–Leibler (KL) divergence.
In addition, $\mathcal{L}^{\sigma}_{\texttt{CEL}}\left(h,g\right) - \mathcal{L}^{\sigma}_{\texttt{CEL}}\left(\mu^{\sigma},g\right)=0$ if and only if $h(b|x)=\mu^{\sigma}(b|x)$ for all $b\in W$ and $x\in \mathcal{X}$. \qed

\section{Proof of Proposition \ref{thm:BNGP}}

Let $\mathcal{G}$ be the set of privacy strategies satisfying Assumption \ref{assp:regular_condition}.
Consider a function $\mathtt{n}:\mathcal{G}\mapsto \mathbb{R}$ that satisfies the following:
\begin{itemize}
    \item If $\mathtt{n}\left(G_{1}\right)\leq \mathtt{n}\left(G_{2}\right)$, then 
    \[
    \mathtt{UtilityL}(G_{1}) \leq \mathtt{UtilityL}(G_{2}).
    \]
    That is, the utility loss is increasing in $\mathtt{n}$.

    \item If $\mathtt{n}\left(G_{1}\right)\leq \mathtt{n}\left(G_{2}\right)$, then 
    \[
    \mathtt{PrivacyL}(H_{1}, G_{1}) \geq \mathtt{PrivacyL}(H_{2}, G_{2}),
    \]
    where $H_{i}$ is the BGP response to $G_{i}$ for $i\in\{1,2\}$. That is, the BGP risk is decreasing in $\mathtt{n}$.
\end{itemize}

We show that the following inequality of privacy loss holds for both problems with privacy budget and utility budget:
\[
\mathtt{PrivacyL}\left(H[G'], G'\right) \geq \mathtt{PrivacyL}\left(H[G^{*}], G^{*}\right).
\]

\paragraph{1. Trade-Off with Privacy Budget}
We first consider the defender's objective function given by (\ref{eq:trade_off_PB}):
\[
\begin{aligned}
    &\mathcal{L}_{\texttt{def}}\left( H,G \middle|Z\right)\\
    &=\begin{cases}
        \mathtt{UtilityL}(G|Z), & \text{ if } \mathtt{PrivacyL}\left( H, G\right)\leq \textup{PB}\\
        \infty,& \text{ othewise}.
    \end{cases}
\end{aligned}
\]
Hence, the BNGP strategy $G^{*}$ is given by:
\begin{equation}\label{eq:tight_prob_1}
    G^{*}\in\arg\min_{G} \mathcal{L}_{\texttt{def}}\left(H[G],G \middle|Z\right), 
\end{equation}
where $H[G]$ is the BNP response to $G$.
Equivalently, we require $\mathtt{n}(G)$ to be at least some threshold, $\mathtt{n}_{\tau}$, for the true BGP risk constraint:
\[
\begin{aligned}
    \left\{G: \mathtt{PrivacyL}\left( H[G], G\right) \leq \textup{PB}\right\} = \left\{G: \mathtt{n}(G)\geq \mathtt{n}_{\tau}\right\}.
\end{aligned}
\]
That is, for the problem (\ref{eq:trade_off_PB}), the defender aims to pick a $G$ to minimize $\mathtt{UtilityL}(G|Z)$ among $\left\{G: \mathtt{n}(G)\geq \mathtt{n}_{\tau}\right\}$.
Since the utility loss is increasing in $\mathtt{n}$, we have that one of the solutions to (\ref{eq:trade_off_PB}) is $G^{*}=G_{\mathtt{n}^{*}_{\tau}}$, where $\mathtt{n}(G_{\mathtt{n}^{*}_{\tau}}) = \mathtt{n}^{*}_{\tau}$.

The alternative $G$ is given by:
\begin{equation}\label{eq:losser_prob_1}
    G'\in\arg\min_{G} \mathcal{L}_{\texttt{def}}\left(\widehat{H}[G],G \middle|Z\right), 
\end{equation}
where $\widehat{H}[G]\in \arg\min\nolimits_{H}\widetilde{\mathcal{L}}^{\gamma,\sigma}_{\texttt{att}}(H,G')$.
Similarly, this is equivalent to find $G$ from
\[
\begin{aligned}
    \left\{G: \mathtt{PrivacyL}\left( \widehat{H}[G], G\right) \leq \textup{PB}\right\} = \left\{G: \mathtt{n}(G)\geq \mathtt{n}_{\tau}\right\}.
\end{aligned}
\]
Thus, there is a threshold $\mathtt{n}'_{\tau}$ such that $G' = G_{\mathtt{n}'_{\tau}}$ where $n(G_{\mathtt{n}'_{\tau}}) = \mathtt{n}'_{\tau}$.

By Lemma \ref{lemma:unique_posterior}, we have, for all $G\in\mathcal{G}$,
\[
\mathtt{PrivacyL}\left(H[G'], G'\right)\geq \mathtt{PrivacyL}(\widehat{H}[G'], G').
\]
Thus, for a given privacy loss budget $\text{PB}$, the constraint in (\ref{eq:losser_prob_1}) is looser than (\ref{eq:tight_prob_1}), 
Then, we have $\mathtt{n}'_{\tau} \leq \mathtt{n}^{*}_{\tau}$.
Since $G^{*}=G_{\mathtt{n}^{*}_{\tau}}$ and $G' = G_{\mathtt{n}'_{\tau}}$ with $\mathtt{n}'_{\tau} \leq \mathtt{n}^{*}_{\tau}$, we have
\[
\mathtt{PrivacyL}\left(H[G'], G'\right) \geq \mathtt{PrivacyL}\left(H[G^{*}], G^{*}\right).
\]

\paragraph{Trade-Off with Utility Budget}
For the optimal privacy-utility trade-off with a utility budget, the defender's objective function is given by (\ref{eq:trade_off_UB}):
\[
\begin{aligned}
    &\mathcal{L}_{\texttt{def}}\left( H,G \middle|Z\right)\\
    &=\begin{cases}
        \mathtt{PrivacyL}\left(H, G\right), & \text{ if } \mathtt{UtilityL}(G|Z) \leq \textup{UB}\\
        \infty,& \text{ othewise}.
    \end{cases}
\end{aligned}
\]
Similar to the problem with a privacy budget, we write down
\begin{equation}\label{eq:tight_prob_2}
    G^{*}\in\arg\min_{G} \mathcal{L}_{\texttt{def}}\left(H[G],G \middle|Z\right), 
\end{equation}
where $H[G]$ is the BNP response to $G$
and
\begin{equation}\label{eq:losser_prob_2}
    G'\in\arg\min_{G} \mathcal{L}_{\texttt{def}}\left(\widehat{H}[G],G \middle|Z\right), 
\end{equation}
where $\widehat{H}[G]\in \arg\min\nolimits_{H}\widetilde{\mathcal{L}}^{\gamma,\sigma}_{\texttt{att}}(H,G')$.
Similarly, for both problem, we require $\mathtt{n}(G)$ to be at lease some threshold $\mathtt{n}_{\tau}$ for the utility constraint:
\[
\begin{aligned}
    \left\{G: \mathtt{UtilityL}\left(G\right) \leq \textup{UB}\right\} = \left\{G: \mathtt{n}(G)\leq \mathtt{n}_{\tau}\right\}.
\end{aligned}
\]
Since the utility loss is independent of $H$, the maximum $\mathtt{n}_{\tau}$ can be attained by (\ref{eq:tight_prob_2}) and (\ref{eq:losser_prob_2}) are the same.
Thus, $\mathtt{PrivacyL}\left(H[G^{*}], G\right)= \mathtt{PrivacyL}(\widehat{H}[G'], G')$.   \qed

\section{Proof of Proposition \ref{prop:post_processing}}

By by Theorem 2.10. of \cite{dong2021gaussian}  (also see \cite{blackwell1951comparison}), we have that for a fixed significance level, the minimum false positive rates (of inferring each individual $k$'s membership status), denoted by $T(G)$ and $T(\textup{Proc}\circ G)$, can be achieved by $G$ and $\textup{Proc}\circ G$ satisfy
\[
T(\textup{Proc}\circ G) \geq T(G).
\]
Thus, $G$ is more informative than $\textup{Proc}\circ G$ according to Blackwell's ordering of informativeness \cite{blackwell1951comparison}.

Next, we show that Blackwell's ordering of information implies the order using conditional entropy.
Let $\rho_{G}: W \mapsto \Delta(\mathcal{X})$ be the conditional density represented by $G$.
Let $P_{G}(x)\equiv \sum_{b}\rho_{G}(x|b)\sigma(b)$.
For $G' = \textup{Proc}\circ G$, let 
\[
\begin{aligned}
    P_{G'}(x') &= \sum_{b}\rho_{G'}(x|b)\sigma(b)\\
    &=\sum_{b}\sigma(b)\int_{x}\textup{Proc}_{G}(x'|x)\rho_{g}(x|b)dx,
\end{aligned}
\]
where $\textup{Proc}_{G}(x'|x)$ captures the process processing from $x$ to $x'$.
Equivalently, we have
\[
P_{G'}(x') = \int_{x'}\textup{Proc}_{G}(x'|x)P_{G}(x)dx.
\]

For a given $G$, the mutual information between membership vector $B$ and output $X$ is defined by
\[
\mathcal{I}_{G}(B;Y) = \mathcal{H}(B) - \mathcal{H}_{G}(B|Y).
\]
Since $\sigma$ is fixed, the entropy $\mathcal{H}(B)$ is a constant.
Therefore, a larger mutual information $\mathcal{I}_{G}(B;Y)$ implies a smaller conditional entropy $\mathcal{H}_{G}(B|Y)$, and vice versa.

Since $G' = \textup{Proc}\circ G$, the data processing inequality implies that
\[
\mathcal{I}_{G'}(B;Y)\leq \mathcal{I}_{G}(B;Y),
\]
which implies
\[
\mathcal{H}_{G}(B|Y)\leq \mathcal{H}(B) - \mathcal{H}_{G'}(B|Y).
\]
By Lemma \ref{lemma:unique_posterior}, $\min_{H}\mathcal{L}^{\sigma}_{\texttt{CEL}}(H,G)$ is the conditional entropy.
Thus, we have $\mathcal{L}^{\sigma}_{\texttt{CEL}}(H^{*},G^{*})\leq \mathcal{L}^{\sigma}_{\texttt{CEL}}(H',G')$, where $G' = \textup{Proc}\circ G$ and $H'$ is the BGP response to $G'$. \qed

\section{Proof of Theorem \ref{prop:composition_risk} }

We first prove the first equality: $\mathcal{L}^{\sigma}_{\texttt{CEL}}(\vec{H}^{*},\vec{G})= \sum^{n}\nolimits_{j=1}\mathcal{L}^{\sigma}_{\texttt{CEL}}(H^{*}_{j},G_{j}) + \Lambda^{\sigma}(\vec{G})$.
For ease of exposition, we focus on the case when there are two mechanisms that are composed. 
That is, $\vec{G}=(G_{1}, G_{2})$.
The proof can be easily extended to general $n\geq 2$.
Let $\rho:\mathcal{D}\mapsto\Delta(\mathcal{X}_{1}\times \mathcal{X}_{2})$ be the underlying joint probability of $\mathcal{M}(\cdot;\vec{G})$ with $\rho_{1}:\mathcal{D}\mapsto\Delta(\mathcal{X}_{1})$ and $\rho_{2}:\mathcal{D}\mapsto\Delta(\mathcal{X}_{2})$ as the underlying marginal probability distributions of the mechanisms $\mathcal{M}(\cdot;G_{1})$ and $\mathcal{M}(\cdot;G_{2})$.
When mechanisms are independent, then $\rho(\cdot|D)=\rho_{1}(\cdot|D)\rho_{2}(\cdot|D)$.
For clarity, we use $\rho(\cdot|b)=\rho(\cdot|D)$ and $\rho_{j}(\cdot|b)=\rho_{j}(\cdot|D)$ for $j\in\{0,1\}$ when $b$ is the membership vector of $D$.

\subsection*{Independent Mechanisms}

Lemma \ref{lemma:unique_posterior} can be easily generalized to multiple mechanisms.
That is, $H^{*}$ represents the joint posterior induced by $\rho$ and $\sigma$, denoted by $h$, and each $H_{j}$ represents the posterior induced by $\rho_{j}$ and $\sigma$ for $j\in\{1,2\}$.
Then, 
\[
\begin{aligned}
    &\mathcal{L}_{\texttt{CEL}}(H^{*}, \vec{G}) \\
    &= -\sum_{b\in W} \int_{\mathcal{X}_{1}\times \mathcal{X}_{2}} \sigma(b)\log(h(b|\vec{x}))\theta(b)\rho(\vec{x}|b)d\vec{x}\\
    =&-\sum_{b}\sigma(b)\log(\theta(b))\cdot \int_{\mathcal{X}_{1}\times \mathcal{X}_{2}}  \rho(\vec{x}|b)d\vec{x}\\
    &- \sum_{b}\int_{\mathcal{X}_{1}\times \mathcal{X}_{2}} \theta(b)\rho(\vec{x}|b) \cdot \log\left(\rho(\vec{x}|b)\right)d\vec{x}\\
    &+\sum_{b}\int_{\mathcal{X}_{1}\times \mathcal{X}_{2}}\theta(b)\rho(\vec{x}|b) \cdot \log\left(\sum_{b'}\rho(\vec{x}|b')\theta(b')\right)d\vec{x}\\
    =& -\sum_{b}\theta(b)\log(\theta(b)) \\
    &- \sum_{b}\int_{\mathcal{X}_{1}\times \mathcal{X}_{2}} \theta(b)\rho(\vec{x}|b) \cdot \log\left(\rho(\vec{x}|b)\right)d\vec{x}\\
    &+ \int_{\mathcal{X}_{1}\times \mathcal{X}_{2}}\left(\sum_{b'}\rho(\vec{x}|b')\theta(b')\right)\cdot \log\left(\sum_{b'}\rho(\vec{x}|b')\theta(b')\right).
\end{aligned}
\]
Similarly, for each $G_{j}$, we have
\[
\begin{aligned}
    &\mathcal{L}_{\texttt{CEL}}(H^{*}_{j}, G_{j}) = -\sum_{b}\theta(b)\log(\theta(b)) \\
    &- \sum_{b}\int_{\mathcal{X}_{j}} \theta(b)\rho_{j}(x_{j}|b) \cdot \log\left(\rho_{j}(x_{j}|b)\right)d x_{j}\\
    &+ \int_{\mathcal{X}_{j}}\left(\sum_{b'}\rho_{j}(x_{j}|b')\theta(b')\right)\cdot \log\left(\sum_{b'}\rho_{j}(x_{j}|b')\theta(b')\right)dx_{j}. 
\end{aligned}
\]
Summing individual losses yields
\[
\begin{aligned}
    &\sum^{2}_{j=1} \mathcal{L}_{\texttt{CEL}}(H^{*}_{j}, G_{j}) =-2\sum_{b}\theta(b)\log(\theta(b))\\
    &- \sum^{2}_{j=1}\sum_{b} \int_{\mathcal{X}_{j}} \theta(b)\cdot \rho_{j}(x_{j}|b)\cdot \log\left(\rho_{j}(x_{j}|b)\right)dx_{j}\\
    &+\sum^{2}_{j=1}\int_{\mathcal{X}_{j}}\left(\sum_{b'}\rho_{j}(x_{j}|b')\theta(b')\right)\\
    &\cdot \log\left(\sum_{b'}\rho_{j}(x_{j}|b')\theta(b')\right)dx_{j}. 
\end{aligned}
\]
Denote the following:
\begin{itemize}
    \item $\mathtt{I} = \sum_{b}\theta(b)\log(\theta(b))$,

    \item $\mathtt{J}_{j} = \sum_{b} \int_{\mathcal{X}_{j}} \theta(b)\cdot \rho_{j}(x_{j}|b)\cdot \log\left(\rho_{j}(x_{j}|b)\right)dx_{j}$,

    \item $\mathtt{K} = \sum_{b}\int_{\mathcal{X}_{1}\times \mathcal{X}_{2}}\theta(b)\rho(\vec{x}|b) \cdot \log\left(\sum_{b'}\rho(\vec{x}|b')\theta(b')\right)d\vec{x}$.
\end{itemize}

Then, we have
\[
\begin{aligned}
    \mathcal{L}_{\texttt{CEL}}(H^{*}, \vec{G}) =& \left(\sum^{2}_{j=1}\left(\mathcal{L}_{\texttt{CEL}}(H^{*}_{j}, G_{j}) -\mathtt{J}_{j}\right) + 2 \mathtt{I}\right) - \mathtt{I} + \mathtt{K}\\ =&\sum^{2}_{j=1}\mathcal{L}_{\texttt{CEL}}(H^{*}_{j}, G_{j}) + \mathcal{I} - \sum^{2}_{j=1} \mathtt{J}_{j} +\mathtt{K}. 
\end{aligned}
\]

By combining $\Lambda^{\sigma}(\vec{G}) = -\left( \mathcal{I} - \sum^{2}_{j=1} \mathtt{J}_{j} +\mathtt{K}\right)$, we have
\[
\begin{aligned}
    \Lambda^{\sigma}(\vec{G}) = -\sum_{b}\theta(b)\int_{\mathcal{X}_{1}\times \mathcal{X}_{2}}\rho(\vec{x}|b)\cdot \log\left(\sum_{b'} \rho(\vec{x}|b')\theta(b')\right)d\vec{x}.
\end{aligned}
\]
For general $n\geq 2$, we have $\mathcal{H}(Q)$, where $Q(\vec{x}) = \sum_{b} \vec{\rho}(\vec{x}|b)\sigma(b)$, $\mathcal{H}(\cdot)$ is the differential entropy.

\subsection*{Dependent Mechanisms}

Next, we proceed with the proof when the mechanisms are dependent.
Similar to the case of independent mechanisms, from Lemma \ref{lemma:unique_posterior}, $H^{*}$ represents the joint posterior induced by $\rho$ and $\sigma$, denoted by $h$, and each $H_{j}$ represents the posterior induced by $\rho_{j}$ and $\sigma$ for $j\in\{1,2\}$.

The BGP risk $\mathcal{L}^{\sigma}_{\texttt{CEL}}(H^{*}, \vec{G})$ is
\[
\begin{aligned}
    \mathcal{L}^{\sigma}_{\texttt{CEL}}(H^{*}, \vec{G}) = -\sum_{b} \int_{\mathcal{X}_{1}\times \mathcal{X}_{2}}\theta(b)\log\left(h(b|\vec{x})\right)\rho(\vec{x}|b)d\vec{x}.
\end{aligned}
\]
Substituting $h(b|\vec{x})$ by Bayes' rule into the loss yields
\[
\begin{aligned}
    &\mathcal{L}^{\sigma}_{\texttt{CEL}}(H^{*}, \vec{G}) \\
    &= -\sum_{b}\int_{\mathcal{X}_{1}\times \mathcal{X}_{2}}\theta(b)\rho(\vec{x}|b) \cdot\log\left(\frac{\rho(\vec{x}|b)\theta(b)  }{ \sum_{b'} \rho(\vec{x}|b')\theta(b') } \right)d\vec{x},
\end{aligned}
\]
which can be broken into three terms:
\[
\begin{aligned}
    &\mathcal{L}^{\sigma}_{\texttt{CEL}}(H^{*}, \vec{G}) = -\sum_{b}\theta(b)\log\left(\theta(b)\right)\cdot \int_{\mathcal{X}_{1}\times \mathcal{X}_{2} } \rho(\vec{x}|b)d\vec{x}\\
    &-\sum_{b} \int_{ \mathcal{X}_{1}\times \mathcal{X}_{2}}\theta(b)\cdot\rho(\vec{x}|b)\cdot \log\left(\rho(\vec{x}|b) \right)d\vec{x}\\
    &+ \int_{ \mathcal{X}_{1}\times \mathcal{X}_{2}} \sum_{b}\theta(b)\rho(\vec{x}|b)\cdot \log\left(\sum_{b'}\rho(\vec{x}|b')\theta(b') \right)d\vec{x}.
\end{aligned}
\]
Denote the following:
\begin{itemize}
    \item $\mathtt{I} = -\sum_{b}\theta(b)\log(\theta(b))$,

    \item $\mathtt{J} = -\sum_{b}\int_{\mathcal{X}_{1}\times \mathcal{X}_{2}}\theta(b) \rho(\vec{x}|b)\cdot \log\left(\rho(\vec{x}|b)\right)d\vec{x}$,

    \item $\mathtt{K}=\int_{\mathcal{X}_{1}\times \mathcal{X}_{2}} \sum_{b} \theta(b) \rho(\vec{x}|b) \cdot \log\left(\sum_{b'}\rho(\vec{x}|b')\theta(b')\right)d\vec{x}$.
\end{itemize}
Thus, the loss can be expressed as:
\[
\mathcal{L}^{\sigma}_{\texttt{CEL}}(H^{*}, \vec{G}) = \mathtt{I} + \mathtt{J} - \mathtt{K}.
\]
By combining $\Lambda^{\sigma}(\vec{G}) = -\left(\mathtt{I}- \mathtt{J} + \mathtt{K}\right)$, we have $\Lambda^{\sigma}(\vec{G}) = \texttt{D}_{\texttt{KL}}\left(Q\| P\right)$, where $Q(\vec{x}) = \sum_{b} \vec{\rho}(\vec{x}|b)\sigma(b)$, $
P(\vec{x}) = \prod_{j=1}^n Q_j(x_j)$ with $Q_j(x_j) = \int_{\vec{\mathcal X}{-j}} Q(x_j, \vec{x}_{-j})\,d\vec{x}_{-j}$, $\mathcal{H}(\cdot)$ is the differential entropy, and $\texttt{D}_{\texttt{KL}}(\cdot)$ is the Kullback–Leibler (KL) divergence.

\subsection*{Second Inequality}

Next, we prove the second inequality of 
\begin{align*}
    &\mathcal{L}^{\sigma}_{\texttt{CEL}}(\vec{H}^{*},\vec{G})= \sum^{n}\nolimits_{j=1}\mathcal{L}^{\sigma}_{\texttt{CEL}}(H^{*}_{j},G_{j}) + \Lambda^{\sigma}(\vec{G})\\
    &\leq \frac{1}{2}\mathbb{E}^{\vec{r}\sim \mathcal{U}^{n}}_{b\sim\sigma}\left[\log\Bigl((2\pi e)^K \det(\Sigma[\vec{G}\left(b,\vec{r}\right)])\Bigr)\right].
\end{align*}

Among all continuous probability density functions on $\mathbb{R}^{K}$ with a given covariance matrix $\Sigma$, the Gaussian distribution maximizes the differential entropy.
In particular, if $V\sim \mathcal{N}(\mu, \Sigma)$, then
\[
\begin{aligned}
    \widehat{\mathcal{H}}(V) &= -\int_{\mathbb{R}^{K}}\phi(v;\mu,\Sigma)\,\log \phi(v;\mu,\Sigma)\,dv \\
    &= \frac{1}{2}\,\log\Bigl((2\pi e)^K\,\det(\Sigma)\Bigr),
\end{aligned}
\]
where $\widehat{\mathcal{H}}$ is the differential entropy $\phi$ is the PDF. This is a standard result in information theory.
For any discrete distribution $\digamma(b)=(\digamma_{k}(b_{k}))^{K}_{k=1}$ that has covariance matrix $\Sigma$, the maximum possible Shannon entropy, $\mathcal{H}^{*}(B)$ satisfies
\[
\mathcal{H}^{*}(B)\leq \frac{1}{2} \log\left(\left(2\pi e\right)^{K}\text{det}\left(\Sigma\right)\right) - \Delta(\Sigma),
\]
where $\Delta(\Sigma)\geq 0$ is a correction term that depends on the support $W$ and the covariance $\Sigma$.
This can be extended to the condition probability $h(b|x)$ which is the probability represented by $\vec{H}^{*}$.

For each fixed $x$, let $\Sigma_{B|X}(x)$ be the covariance of $h(b|x)$.
Then, we have
\[
\begin{aligned}
    &\mathcal{H}(h(\cdot|x)) = -\sum_{b}\log(h(b|x))\\& \leq \frac{1}{2}\log\left((2\pi e)^{K} \text{det}\left(\Sigma_{X|Y}(y)\right)\right) - \Delta(\Sigma_{X|Y}(y)).
\end{aligned}
\]
Since the BGP risk is the conditional entropy, we have
\[
\begin{aligned}
    &\mathcal{L}^{\sigma}_{\text{CEL}}(\vec{H}^{*}, \vec{G})=\mathbb{E}_{X}\left[\mathcal{H}(h(\cdot|x)) \right].
\end{aligned}
\]
Thus, 
\[
\begin{aligned}
    \mathcal{L}^{\sigma}_{\text{CEL}}(\vec{H}^{*}, \vec{G})&\leq \mathbb{E}^{\vec{r}\sim \mathcal{U}^{n}}_{b\sim\sigma}\Big[\frac{1}{2}\log\Bigl((2\pi e)^K \det(\Sigma[\vec{G}\left(b,\vec{r}\right)])\Bigr) \\&- \Delta(\Sigma_{X|Y}(y))\Big].
\end{aligned}
\]
Since $\Delta(\Sigma_{X|Y}(y))$, we obtain the second inequality.

\section{Proof of Proposition \ref{prop:static_subjective_prior}}

Proposition \ref{prop:static_subjective_prior} directly follows Lemma \ref{lemma:unique_posterior}. 
That is, we have that $\mathcal{L}^{\theta}_{\texttt{CEL}}(H^{*},G)$ achieves the minimum when $H^{*}$ represents the posterior distribution of the distribution represented by $G$.
That is, for all $H$,
\[
\mathcal{L}^{\theta}_{\texttt{CEL}}(H^{*}, G)\leq \mathcal{L}^{\theta}_{\texttt{CEL}}(H, G).
\]
Thus, it also holds for $H=H^{*}_{\sigma}$. \qed

\section{Proof of Proposition \ref{prop:Bayes_plausibility} }

Let $\rho_{G}:W\mapsto \Delta(\mathcal{X})$ denote the conditional probability given $G$.
Let $h_{\theta}:\mathcal{X}\mapsto \Delta(\{0,1\})^{K}$ denote the posteriors using $\theta$.
Let $x_{G}$ and $x_{q}$ denote the outputs by $G$ and $q$, respectively.
Given $q$, the refined posterior using both $q$ and $\rho_{G}$ can be constructed by
\[
\begin{aligned}
    h^{q}_{G}(b|x_{G}, x_{q})=\frac{q(x_{q}|b)g(x_{G}|b)\theta(b)}{\sum_{b'} q(x_{q}|b')g(x_{G}|b')\theta(b')}.
\end{aligned}
\]

For a fixed $x_{G}$, we have
\[
\begin{aligned}
    \mathbb{E}_{x_{q}}\left[\mathcal{H}\left(h^{q}_{G}(\cdot|x_{G}, x_{q})\right) \right]\leq \mathcal{H}\left(h_{\theta}(\cdot|x_{G})\right),
\end{aligned}
\]
where $\mathcal{H}$ is the Shannon entropy.
This is because additional conditioning reduces entropy.
By taking the expectation over $x_{G}$, we have
\[
\begin{aligned}
    \begin{aligned}
        &\mathbb{E}_{x_{G}}\left[\mathbb{E}_{x_{q}}\left[\mathcal{H}\left(h^{q}_{G}(\cdot|x_{G}, x_{q})\right) \right]\right] = \mathcal{L}^{\theta}_{\texttt{CEL}}(H^{*}_{\sigma}, G) \\
        &\leq \mathbb{E}_{x_{G}}\left[\mathcal{H}\left(h_{\theta}(\cdot|x_{G})\right)\right] = \mathcal{L}^{\theta}_{\texttt{CEL}}(H^{*}_{\theta}, G).
    \end{aligned}
\end{aligned}
\]
\qed

\section{Proof of Theorem \ref{thm:ordering_DP_BGP}}

Let $\rho:\mathcal{D}\mapsto \Delta(\mathcal{X})$ and $\rho':\mathcal{D}\mapsto \Delta(\mathcal{X})$ be the probability distributions induced by $G$ and $G'$, respectively.
We use $\rho(\cdot|b)=\rho(\cdot|D)$ and $\rho'(\cdot|b)=\rho'(\cdot|D)$ when the membership vector of $D$ is $b$.
When $G'$ is more Blackwell informative \cite{blackwell1951comparison,de2018blackwell} than $G$, we denote $\rho'\succeq_{\mathtt{B}} \rho$.
We say $g'\succeq_{\mathtt{B}} g'$ is more (Blackwell) informative than $g$ if there exists a stochastic process $\eta(x|x')$ such that $\rho(x|b)=\int_{x'}\eta(x|x')\rho'(x'|b)dx$.
In addition, let $\mathcal{I}^{\sigma}(B;X_{G})$ and $\mathcal{I}^{\sigma}(B;X_{G'})$ denote the mutual information under $G$ and $G'$, respectively, when the prior is $\sigma$.

We start by showing that \textbf{(1)} is equivalent to $\rho'\succeq_{\mathtt{B}} \rho$.

\paragraph{\textup{\textbf{Blackwell Order $\Rightarrow$ (1)}}}
Suppose that $\rho(x|b)=\int_{x'}\eta(x|x')\rho'(x'|b)dx$. By definition, there is a stochastic process $\eta$ such that $\rho(x|b)=\int_{x'}\eta(x|x')\rho'(x'|b)dx$. In other words, we have the Markov chain: 
\[
B\rightarrow X_{G} \rightarrow X_{G'}.
\]
By the data processing inequality, we have
\[
\mathcal{I}^{\sigma}(B;X_{G})\leq \mathcal{I}^{\sigma}(B;X_{G'}),
\]
for \textit{every} prior $\sigma\in \Theta$.

By definition, the mutual information can be expressed in terms of conditional entropy $\mathcal{H}^{\sigma}(B;X)$:
\[
\mathcal{I}^{\sigma}(B;X) = \mathcal{B} - \mathcal{H}^{\sigma}(B;X).
\]
Thus, we have
\[
\mathcal{H}^{\sigma}(B;X_{G'}) \leq \mathcal{H}^{\sigma}(B;X_{G}),
\]
for all $\sigma\in\Theta$.
Then, Lemma \ref{lemma:unique_posterior} implies $\mathcal{L}^{\sigma}{\texttt{CEL}}(G) \geq \mathcal{L}^{\sigma}{\texttt{CEL}}(G')$ for all $\sigma\in\Theta$.

\paragraph{\textup{\textbf{(1) $\Rightarrow$ Blackwell Order}}}
Suppose that we have $\mathcal{L}^{\sigma}{\texttt{CEL}}(G) \geq \mathcal{L}^{\sigma}{\texttt{CEL}}(G')$ for all $\sigma\in\Theta$.
Then, for all $\sigma\in\Theta$, we have
\[
\mathcal{I}^{\sigma}(B;X) = \mathcal{H}(B) - \mathcal{H}^{\sigma}(B;X),
\]
which implies $\mathcal{I}^{\sigma}(B;X_{G})\leq \mathcal{I}^{\sigma}(B;X_{G'})$.

The conditional entropy can be derived from the \textit{Bayes risk} under the logarithmic loss $\mathtt{L}(b,h|x) = -\log(h(b|x))$:
\[
\min_{h}\mathbb{E}^{\sigma}\left[-\log(h(B|X))\right] = \mathcal{H}^{\sigma}\left(B|X\right).
\]
So, If $\mathcal{H}^{\sigma}\left(B|X_{G}\right)\geq \mathcal{H}^{\sigma}\left(B|X_{G'}\right)$ for all $\sigma$, then the privacy strategy $G'$ yields a lower Bayes risk than $G$ under logarithmic loss for all priors.
The logarithmic loss is a strictly proper scoring rule, which means that the $h$ that coincides with the true posterior distribution uniquely minimizes the expected loss.
In addition, the dominance under the logarithmic loss for all priors implies dominance in all decision problems \cite{blackwell1951comparison,blackwell1953equivalent}.
Then, by Blackwell–Sherman–Stein Theorem \cite{blackwell1951comparison,sherman1951theorem,stein1951notes} (also, see Theorem 1 of  \cite{de2018blackwell}), we can obtain $\rho'\succeq_{\mathtt{B}} \rho$.

Newt, we introduce the notion of \textit{$f$-differential privacy} \cite{dong2022gaussian}. 
Let $D$ and $D'$ be any two adjacent datasets (i.e., $D\simeq D'$), differing in individual $k$'s data point. 

\begin{definition}[Trade-off Function (Definition 2.1 \cite{dong2022gaussian})]\label{def:trade_off_function}
    For any two probability distributions $P$ and $Q$ on the same space, define the trade-off function $T(P,Q):[0,1]\mapsto [0,1]$ as
    \[
    T(P,Q)(\alpha) = \inf\left\{\beta_{\phi}:\alpha_{\phi}\leq \alpha\right\},
    \]
    where the infimum is taken over all (measurable) test rules.
\end{definition}

\begin{definition}[$f$-Differential Privacy \cite{dong2022gaussian}]\label{def:f_DP}
    Let $f$ be a trade-off function. A mechanism $\mathcal{M}$ is said to be $f$-differentially private ($f$-DP) if 
    \[
    T(\mathcal{M}(D),\mathcal{M}(D'))\geq f,
    \]
    for all $D\simeq D'$.
\end{definition}

By Lemma \ref{lemma:app_f_DP} (shown later), we know that $\mathcal{M}(\cdot;G)$ using BNGP strategy $G$ is $f$-differentially private ($f$-DP).
Let $f$ and $f'$, respectively, be the corresponding f functions such that $\mathcal{M}(\cdot;G)$ is $f$-DP and $\mathcal{M}(\cdot;G')$ is $f$-DP.
Let $T(\cdot)$ be the trade-off function defined by Definition \ref{def:trade_off_function} in Appendix \ref{app:proof_relationship_DP}.
Thus, for any adjacent datasets $D\simeq D'$, $T(\mathcal{M}(D;G),\mathcal{M}(D';G))$ captures the \textit{minimum} Type-II errors for a chosen significance level $\alpha$.
From Theorem 2 of \cite{dong2022gaussian} (i.e., Theorem 10 of \cite{blackwell1951comparison}), we have, for all significance level $\alpha$,
\[
\begin{aligned}
    &\rho'\succeq_{\mathtt{B}} \rho \Longleftrightarrow \\&T(\mathcal{M}(D;G'),\mathcal{M}(D';G')) \leq T(\mathcal{M}(D;G),\mathcal{M}(D';G)).
\end{aligned}
\]
Thus, we also have that $\rho'\succeq_{\mathtt{B}} \rho$ is equivalent to $f'(\alpha)\leq f(\alpha)$ for all $\alpha$.
The equivalence between 1 and 2 in Theorem 1 of \cite{kaissis2024beyond}, we establish the equivalence of our \textbf{(1)} and \textbf{(2)}.

\paragraph{\textup{\textbf{(1) or (2) $\Rightarrow$ (3)}}}
$\overline{\mathtt{Adv}}^{\gamma}_{\sigma}(G)$ is defined by
\[
\overline{\mathtt{Adv}}^{\gamma}_{\sigma}(G)\equiv \max_{H}\mathtt{Adv}^{\gamma}_{\sigma}(H,G).
\]
Proposition \ref{prop:weighted_membership_advantage} implies that $\max_{H}\mathtt{Adv}^{\gamma}_{\sigma}(H,G)$ is equivalent to 
\[
\min_{H}\mathcal{L}^{\gamma,\sigma}_{\texttt{att}}(H, G).
\]
Thus, by Theorem \cite{kaissis2024beyond}, we obtain that $\rho'\simeq_{\mathtt{B}} \rho$ implies 
\[
\min_{H}\mathcal{L}^{\gamma,\sigma}_{\texttt{att}}(H, G)\geq \min_{H}\mathcal{L}^{\gamma,\sigma}_{\texttt{att}}(H, G'),
\]
which is equivalent to $\overline{\mathtt{Adv}}^{\gamma}_{\sigma}(G) \leq \overline{\mathtt{Adv}}^{\gamma}_{\sigma}(G')$, $\forall \sigma\in\Theta, \gamma\in(0,1]$.  

Since $\rho'\simeq_{\mathtt{B}} \rho$ $\Leftrightarrow$ \textbf{(1)} $\Leftrightarrow$ \textbf{(2)}, we can conclude the proof. 
\qed

\begin{lemma}\label{lemma:app_f_DP}
    Suppose $\mathcal{M}(\cdot;G)$, where $G$ is a BNGP strategy. Then, $\mathcal{M}(\cdot;G)$ satisfies $f$-differential privacy.
\end{lemma}

\begin{proof}

Let $D$ and $D'$ be any two adjacent datasets (i.e., $D\simeq D'$), differing in individual $k$'s data point. 
The attacker in the scenarios of $f$-DP and DP aims to infer whether individual $k$'s data point is used by the mechanism or not. 
This can be formulated as a hypothesis-testing problem:
\begin{center}
    $H_0:$ $b_{k}=0$ versus $H_{1}:$ $b_{k}=1$.
\end{center}

Let $\phi_{k}(x)\in [0,1]$ be a test function (or rejection rule), so that $\phi_{k}(x)$ is the probability of rejecting $H_{0}$ upon observing $x$.
The significance level of the test is then defined by
\[
\alpha_{\phi_{k}} = \mathbb{E}\left[\phi_{k}(X)|b_{k}=0\right] = \int_{\mathcal{X}}\phi_{k}(x)\vec{\rho}(x|D, b_{k}=0)dx,
\]
where the expectation is taken under the null hypothesis $H_{0}$.

Define the likelihood ratio:
\[
\Lambda(x) = \frac{ \rho(x|D, b_{k}=1) }{ \rho(x|D, b_{k}=0) }.
\]
The Neyman-Pearson Lemma \cite{neyman1933ix} states that among all tests $\phi(x)$ with significance level $\alpha_{\phi_{k}}$ at most $\alpha$, the test that maximizes the power
\[
1-\beta^{\alpha}_{\phi_{k}}=\mathbb{E}\left[\phi_{k}(X)|b_{k}=1\right] = \int_{\mathcal{X}}\phi_{k}(x)\rho(x|D, b_{k}=1)dx,
\]
is given by the likelihood ratio test which rejects $H_{0}$ when $\Lambda(x) = \tau$.
That is, the optimal test is given by
\[
\phi^{*}_{k}(x) = \begin{cases}
    1, & \text{ if } \Lambda(x)>\tau,\\
    \mathtt{c}, & \text{ if } \Lambda(x) = \tau,\\
    0, & \text{ if }\Lambda(x)<\tau,
\end{cases}
\]
where $\mathtt{c}$ is chosen (if needed) to ensure that the significance level of the test is exactly $\alpha_{\phi^{*}_{k}}$.
Thus, our hypothesis testing problem has a uniformly most powerful (UMP) test specified by $\phi^{*}_{k}(x)$ for any $D\simeq D'$ differing in individual $k$'s data point, for all $k\in\mathcal{K}$.

For any $D\simeq D'$, define $f(D,D'):[0,1]\mapsto [0,1]$ by
\[
f(D,D')(\alpha)\equiv\inf_{k\in\mathcal{K}} \beta^{\alpha}_{\phi^{*}_{k}}. 
\]
It is clear that there exists a trade-off function for our hypothesis testing: for all $\alpha\in[0,1]$.
\[
f(\alpha) = \inf_{D\simeq D'} f(D,D')(\alpha).
\]
Thus, the mechanism $\mathcal{M}(\cdot;G)$ satisfies $f$-DP.

\end{proof}

\section{Proof of Proposition \ref{prop:error_bound_H}}

The neural network $H_{\lambda_{\mathtt{a}}}: \mathcal{X} \to [0,1]^K$ has $\mathtt{l}$ hidden layers, each of width $\mathtt{w} \geq 7K+1$. Each hidden layer applies the ReLU activation function, which is 1-Lipschitz. The final layer applies the sigmoid function coordinatewise, ensuring that the output values remain in $[0,1]$. Since the sigmoid function is Lipschitz with constant $L_\sigma = 0.25$, the overall Lipschitz constant of the network is
    \begin{equation*}
        L_{\text{network}} = 1^{\mathtt{l}} \cdot 0.25 = 0.25.
    \end{equation*}

By assumption, the posterior $\mu_{\lambda_{\mathtt{d}}}: \mathcal{X} \to \Delta(\{0,1\})^K$ has effective intrinsic dimension $\alpha_{\lambda_{\mathtt{d}}}$. This means that for any $\eta > 0$, there exists a mapping $F: \mathcal{X} \to \mathbb{R}^{\alpha_{\lambda_{\mathtt{d}}}}$ and a Lipschitz mapping $T: \mathbb{R}^{\alpha_{\lambda_{\mathtt{d}}}} \to \Delta(\{0,1\})^K$ such that
\begin{equation*}
    \sup_{x \in \mathcal{X}} \mathcal{W}_1(\mu_G(x), T(F(x))) \leq \eta.
\end{equation*}
This property ensures that the complexity of the target distribution $\mu_G(x)$ is governed by $\alpha_{\lambda_{\mathtt{d}}}$ rather than the ambient dimension $K$.

Standard universal approximation results for ReLU networks imply that a network with width $\mathtt{w}$ and depth $\mathtt{l}$ can approximate a Lipschitz function on a compact domain with an error that scales as the number of linear regions increases. Since a ReLU network of this size can produce up to $O(\mathtt{w}^2\mathtt{l})$ linear regions, there exists a neural network $H_{\lambda_{\mathtt{a}}}$ that approximates $T \circ F$ such that
    \begin{equation*}
        \sup_{x \in \mathcal{X}} \|H_{\lambda_{\mathtt{a}}}(x) - T(F(x))\| \leq C' (\mathtt{w}^{2} \mathtt{l})^{-1/\alpha_{\lambda_{\mathtt{d}}}},
    \end{equation*}
    where $C'$ depends on the intrinsic properties of $T \circ F$ but not on $\mathtt{w}$ or $\mathtt{l}$.

The final step is to convert this uniform approximation error into a Wasserstein-1 error bound. By the Lipschitz property of $H_{\lambda_{\mathtt{a}}}$, we obtain
\begin{equation*}
     \mathcal{W}_1(T(F(x)), H_{\lambda_{\mathtt{a}}}(x)) \leq 0.25 C' (\mathtt{w}^{2} \mathtt{l})^{-1/\alpha_{\lambda_{\mathtt{d}}}}.
\end{equation*}
Using the triangle inequality for the Wasserstein-1 metric,
\begin{equation*}
     \begin{aligned}
         &\mathcal{W}_1(\mu_G(x), H_{\lambda_{\mathtt{a}}}(x)) \\&\leq \mathcal{W}_1(\mu_G(x), T(F(x))) + \mathcal{W}_1(T(F(x)), H_{\lambda_{\mathtt{a}}}(x)).
     \end{aligned}
\end{equation*}
Since $T(F(x))$ approximates $\mu_G(x)$ arbitrarily well, the first term can be absorbed into the constant. Thus, we conclude that
\begin{equation*}
     \mathcal{W}_1(\mu_G(x), H_{\lambda_{\mathtt{a}}}(x)) \leq 0.25 \mathtt{C} (\mathtt{w}^{2} \mathtt{l})^{-1/\alpha_{\lambda_{\mathtt{d}}}},
\end{equation*}
where $\mathtt{C}$ is independent of $\mathtt{w}$ and $\mathtt{l}$ but depends on the intrinsic properties of $\mu_{\lambda_{\mathtt{d}}}$ and the low-dimensional representation.
\qed

\section{Experiment Details}\label{sec:app_Hyperparameters}

\subsection{Dataset}

Our experiments use three datasets: \textit{Adult dataset} (UCI Machine
Learning Repository), \textit{MNIST dataset}, and \textit{genomic dataset}.
The genomic dataset was provided by the 2016 iDASH Workshop on Privacy and Security \cite{tang2016idash}, derived from the 1000 Genomes Project \cite{10002015global}.
The genomic dataset used in our experiments was initially provided by the organizers of the 2016 iDash Privacy and Security Workshop \cite{tang2016idash} as part of their challenge on Practical Protection of Genomic Data Sharing Through Beacon Services. 
In this research, we follow \cite{venkatesaramani2021defending,venkatesaramani2023enabling} and employ SNVs from chromosome $10$ for a subset of $400$ individuals to construct the Beacon, with another $400$ individuals excluded from the Beacon. 
We use $800$ individuals with different numbers of SNVs of each individual on Chromosome $10$.
In the experiments, we randomly select $400$ individuals from the $800$ to constitute a dataset according to the uniform distribution. 
The experiments were conducted using an NVIDIA A40 48G GPU. PyTorch was used as the deep learning framework.

\subsection{Notes: Experiment Details for LRTs}

\begin{figure}[htb]
    \centering
    \begin{subfigure}[b]{0.32\textwidth}
        \includegraphics[width=\textwidth]{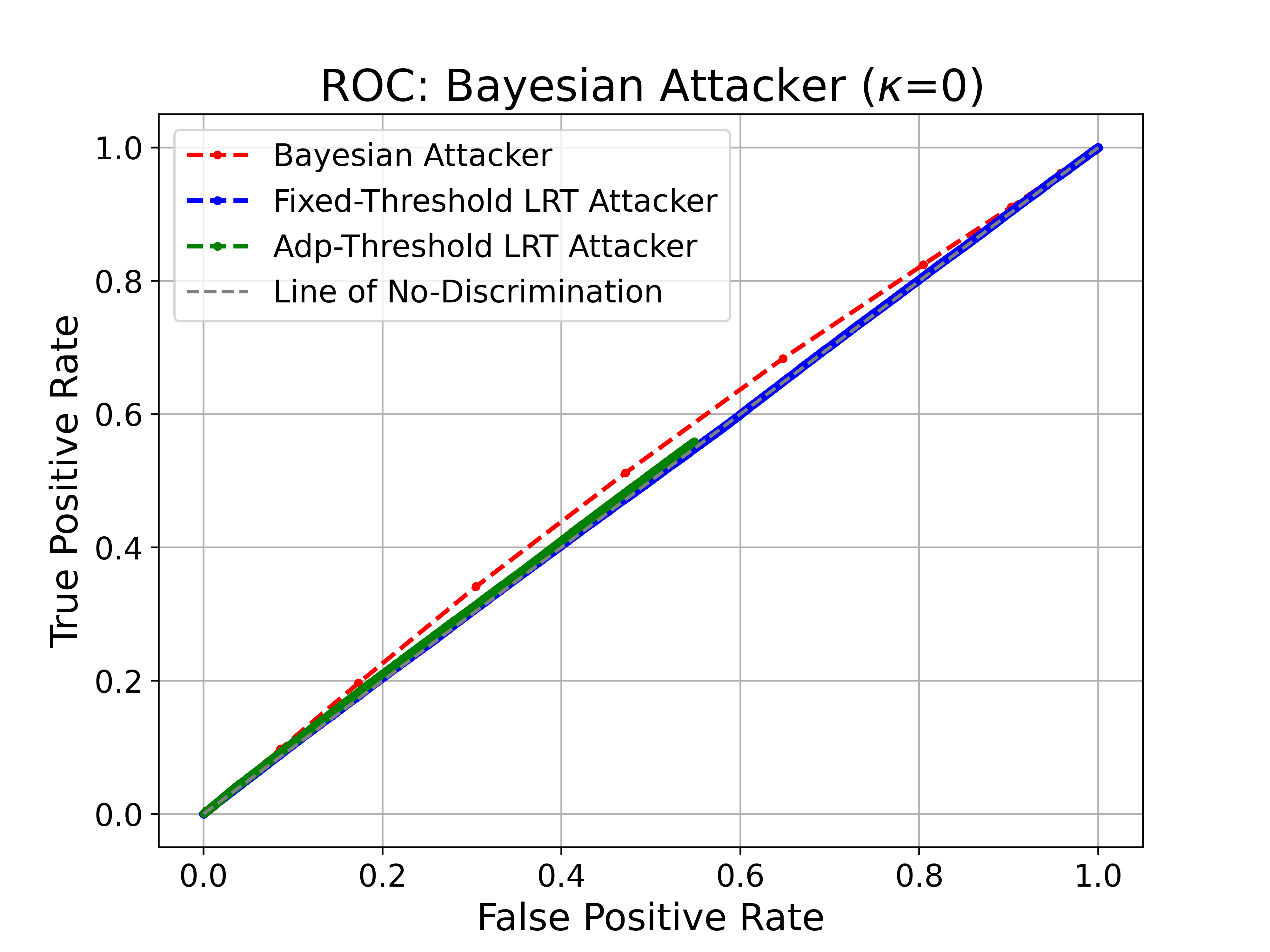}
        \caption{}
        \label{fig:fig1ap}
    \end{subfigure}
    \hfill 
    \begin{subfigure}[b]{0.32\textwidth}
        \includegraphics[width=\textwidth]{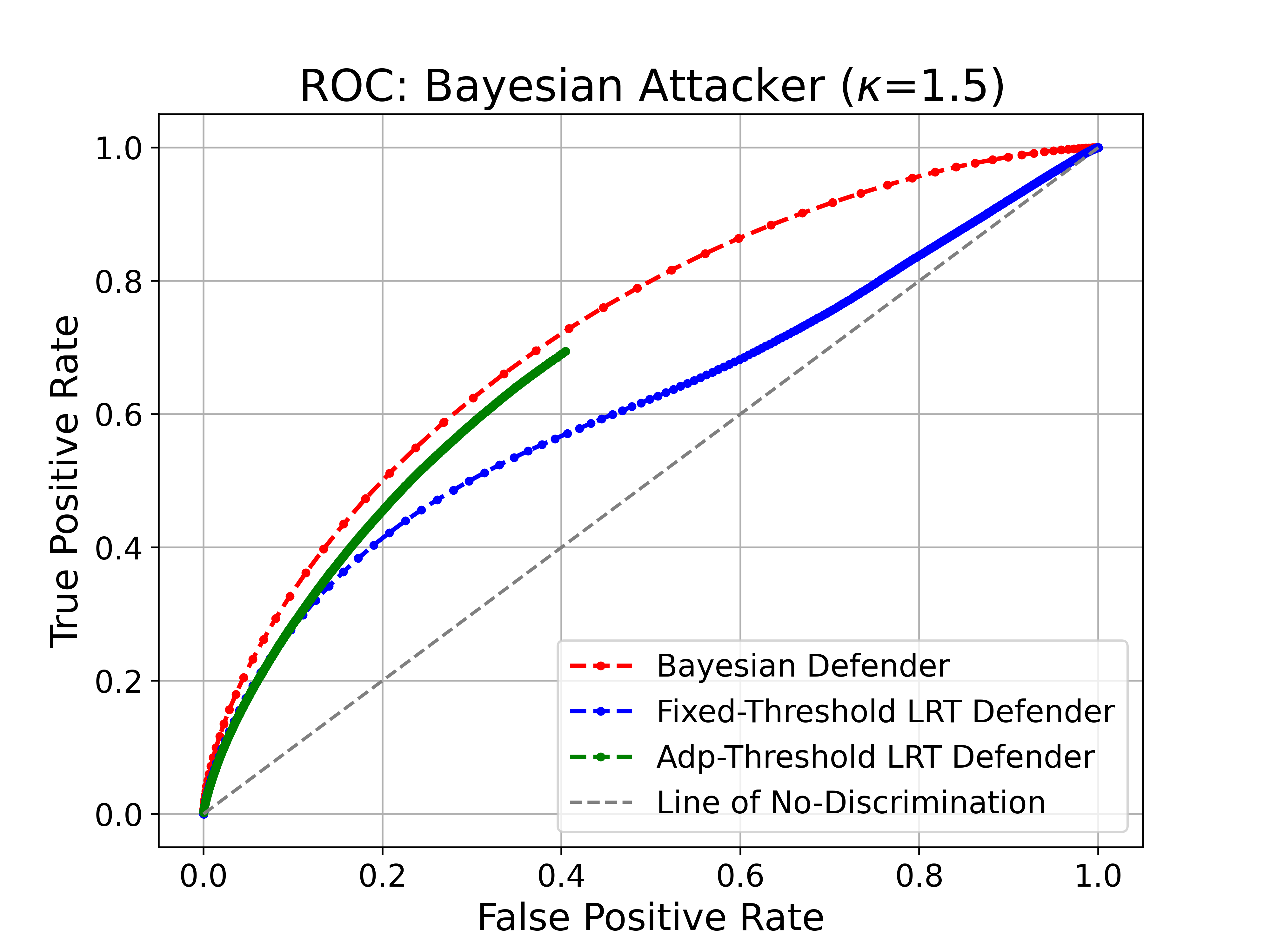}
        \caption{}
        \label{fig:fig2ap}
    \end{subfigure}
    \hfill
    \begin{subfigure}[b]{0.32\textwidth}
        \includegraphics[width=\textwidth]{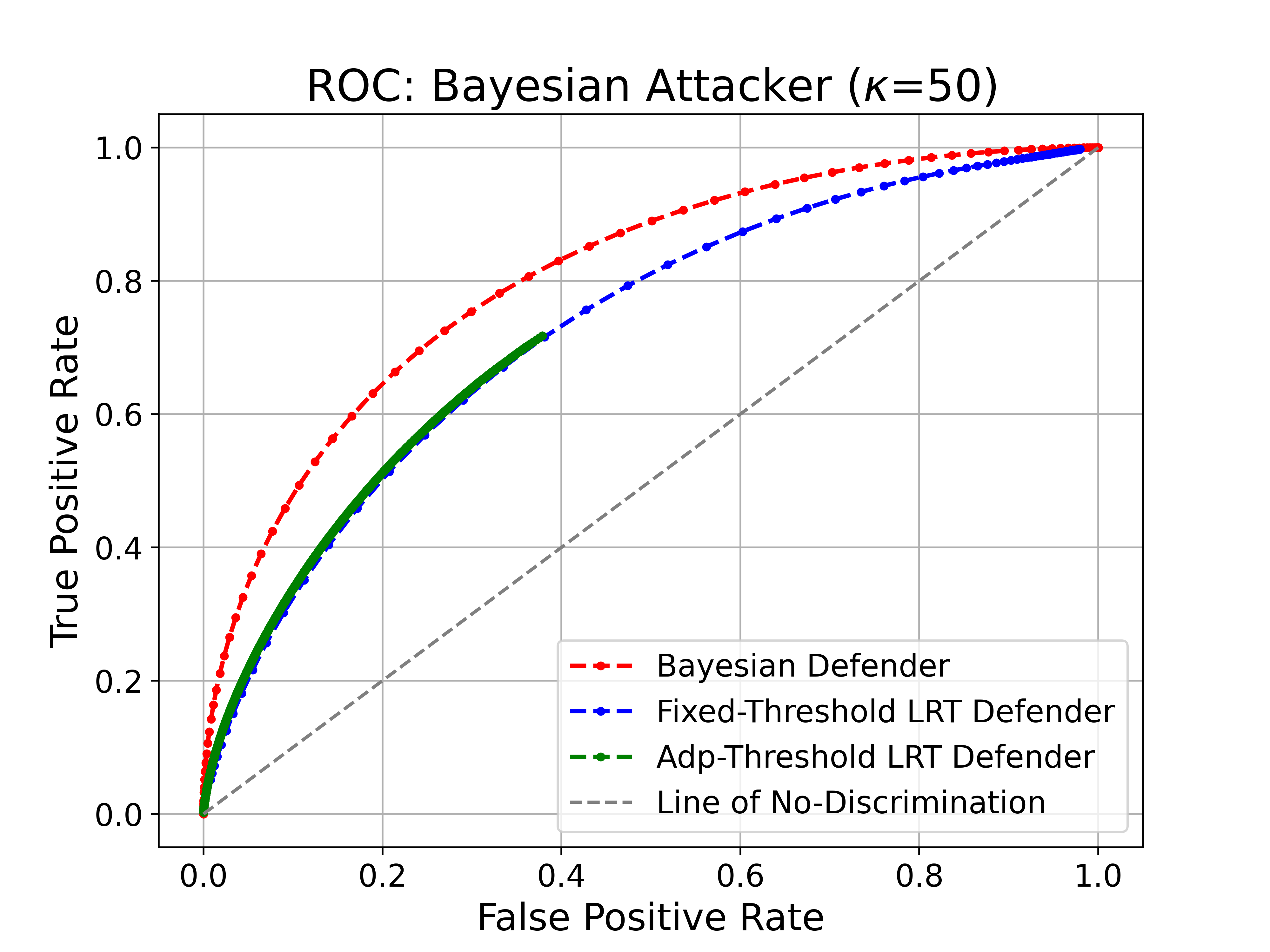}
        \caption{}
        \label{fig:fig3ap}
    \end{subfigure}
    \begin{subfigure}[b]{0.32\textwidth}
        \includegraphics[width=\textwidth]{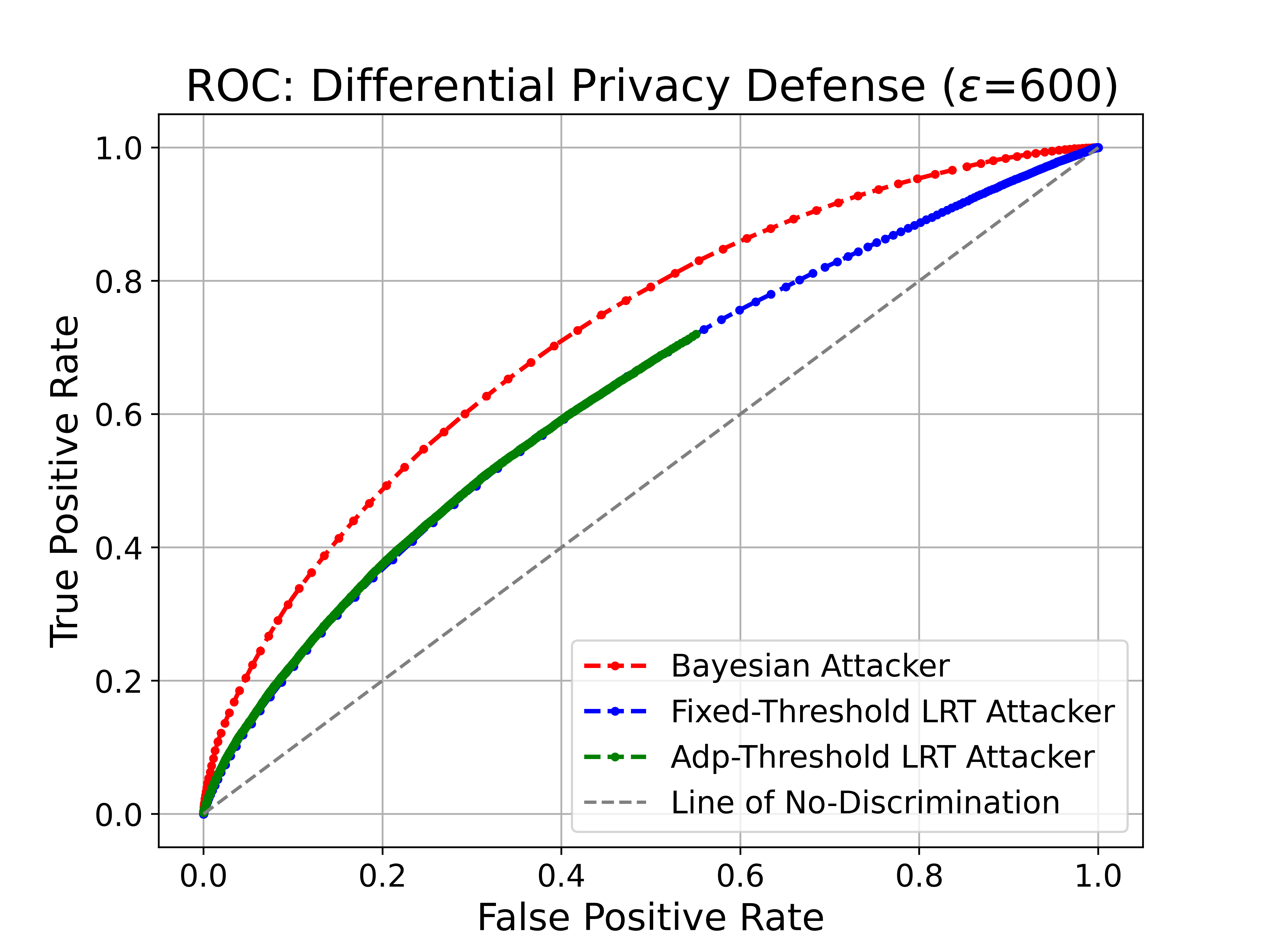}
        \caption{}
        \label{fig:fig4ap}
    \end{subfigure}
    \caption{(a)-(c): Bayesian Defender with $\kappa=0, 1.5, 50$, respectively. (d): Different attackers under non-strategic DP with $\epsilon=600$.}
    \label{fig:six_figuresap}
\end{figure}

The output of the defender's neural network \( G_{\lambda_{D}} \) is a noise term within the range \([-0.5, 0.5]\). We assess the strength of privacy protection using the attacker's ROC curve, converting \( H_{\lambda_{A}} \)'s output to binary values \( s_{k} \in \{0,1\} \) by varying thresholds. A lower AUC indicates stronger privacy protection by \( G_{\lambda_{D}} \). 
In addition to the proxies from Section \ref{sec:BNGP}, we use the sigmoid function to approximate the threshold-based rejection rule of the LRT. Specifically, \( \mathbf{1}\{\mathtt{lrs}(d_{k}, x)\leq\tau\} \) is approximated by \( 1/(1+\exp(-(\tau-\mathtt{lrs}(d_{k}, x)))) \), where \( \mathtt{lrs}(d_{k}, x) \) is the log-likelihood statistic. Similarly, the sigmoid function approximates \( \mathbf{1}\{\mathtt{lrs}(d_{k}, r)\leq \tau^{(N)}(r)\} \). The fixed- and adaptive-threshold LRT defenders optimally select \( g \) by solving (\ref{eq:nariveLRT}) and (\ref{eq:adaptiveLRT}), as detailed in Appendix \ref{app:existence_lrt_attack}.

\subsection{Notes: Bayesian Defender vs. Differential Privacy}\label{sec:app_naive_DP}
In this experiment (Figure \ref{fig:fig6}), we illustrate the advantages of the Bayesian defender (i.e., using the BNGP strategy) over standard DP in addressing defender-customized objectives for the privacy-utility trade-off, when the same utility loss is maintained in the trade-off of privacy and utility.

In this experiment, we consider a specific loss function for the defender:
\[
\ell_{D}(\delta, b, s) \equiv v(s, b) + \sum\nolimits_{j \in Q} \kappa_{j} |\delta_{j}|,
\]
where \(\kappa_{j} \geq 0\) represents the defender's preference for balancing the privacy-utility trade-off for the summary statistics of the \(j\)-th attribute (e.g., SNV in genomic data). In genomic datasets, each SNP position corresponds to a specific allele at a particular genomic location, and the importance of these positions can vary significantly depending on their association with diseases or traits in medical studies. Consequently, different SNPs may require varying levels of data quality and utility, necessitating less noise for some positions. For SNPs where higher data utility is crucial, we assign larger \(\kappa_{j}\) values to increase the weight of noise costs in the defender's decision-making process. This position-dependent weighting enables a more customized and refined privacy-utility trade-off.

In the experiment, we define \(\vec{\kappa} = (\kappa_{j})_{j \in Q}\) for SNV positions, where \(\kappa_{j} = 0\) for \(90\%\) of the \(5000\) SNVs and \(\kappa_{j} = 50\) for the remaining \(10\%\). The BNGP strategy in this setting results in an average utility loss of \(0.0001\).

The sensitivity of the summary statistics function \(f(\cdot)\) is given by \(\textup{sensitivity} = \frac{m}{K^{\dagger}}\) (see Appendix \ref{sec:app_differential_privacy}), where \(m = |Q|\) and \(1 \leq K^{\dagger} \leq K\) is the number of individuals in \(U\) included in the dataset. For the experiment, the dataset comprises \(400\) individuals, each with \(5000\) SNVs, resulting in a sensitivity of \(\frac{m}{K^{\dagger}} = 12.5\).

The scale parameter of the Laplace distribution in the DP framework is:
\[
\frac{\textup{sensitivity}}{\epsilon}.
\]

To match the utility loss of \(0.0001\) (measured as the expected absolute value of the noise), the scale parameter must equal \(0.0001\). This implies:
\[
\frac{\textup{sensitivity}}{\epsilon} = 0.0001,
\]
which gives \(\epsilon = 1.25 \times 10^{5}\).

In general, the number of SNVs (\(m\)) is often much larger than the number of individuals (\(K\)), i.e., \(m \gg K\). Consequently, small \(\epsilon\) values (e.g., between \(1\) and \(10\)) result in very large scale parameters for the Laplace distribution. Therefore, relatively large \(\epsilon\) values are chosen to preserve the utility of genomic datasets. For example, in \cite{venkatesaramani2023enabling}, the values of \(\epsilon\) are selected from \(\{10,000, 50,000, 100,000, 500,000,\) \(1 \textup{ million}, 5 \textup{ million}, 10 \textup{ million}\}\).

\subsection{Notes: Score-Based Attacker: }

Figure \ref{fig:fig4x} compares attackers under the defense trained against the score-based attacker. The Bayesian attacker significantly outperforms the others, achieving near-perfect classification, while LRT and adaptive LRT perform similarly but lag behind. 
As explained in \cite{dwork2015robust}, the score-based attacker is assumed to have less external information and knowledge than the Bayesian, the fixed-threshold LRT, and the adaptive LRT attackers. 
Theoretically, the score-based attacker uses $\mathcal{O}(n^{2}\log(n))$ SNVs, where $n$ is the number of individuals in the dataset.
In the experiments for Figure \ref{fig:fig4x}, to guarantee certain accuracy for the score-based attacker, we consider $20$ individuals and each time a dataset of $5$ individuals with $4000$ SNVs being randomly sampled. In this setting, the Bayesian attacker performs very well (with AUC close to $1$).

\subsection{Notes: Adult and MNIST Dataset}\label{sec:app:adult_mnist}

\paragraph{Adult Dataset }
In the experiments using the Adult dataset, the original mechanism $f$ (i.e., without privacy protection) releases the summary statistics of the Adult dataset. Specifically, we turn the attributes of the Adult dataset into binary values to simplify the representation of categorical and continuous attributes. For example, categorical attributes like ``occupation" or ``education level" are one-hot encoded, while continuous attributes like ``age" are discretized into binary intervals. This binary transformation allows us to construct a dataset that represents the presence or absence of specific attribute values, making it compatible with our framework for privacy protection and utility optimization.
The summary statistics released include the counts or proportions of individuals possessing specific binary attributes. These statistics form the basis for evaluating the membership inference risks and utility trade-offs in our experiments. By using this transformed representation, we ensure the methodology aligns with the assumptions of our privacy-utility trade-off framework.

\paragraph{MNIST Dataset}
For the MNIST dataset, the original mechanism is a trained classifier that outputs predicted class probabilities for given input images. Specifically, this classifier is trained on the MNIST training set to perform digit recognition, mapping each image to a probability distribution over the 10-digit classes (0 through 9).
In our experiments, we consider the privacy of the test data (or inference dataset) used to query the classifier. The attacker aims to infer whether a specific test image belongs to the inference dataset based on the output probabilities provided by the classifier.

\subsection{Notes: BGBP Response}

The BGBP response acts as a constraint for the \textit{defender} since the defender's choice of \( G \) induces \( H \), which (1) represents the attacker's best response, and (2) satisfies the conditions defined by \( \mathtt{PH}[\vec{G}; \epsilon] \). The attacker, however, simply responds optimally to the defender's choice of \( G \). To incorporate the conditions for \( H \) set by \( \mathtt{PH}[\vec{G}; \epsilon] \), we apply the penalty method to the defender's loss function. 
In our experiments, we relax the strict pure differential privacy framework and focus on a class of neural networks \( G \) that select \( \epsilon \) for a Gaussian distribution \( \mathcal{N}(0, \mathtt{Var}(\epsilon)) \), where \( \mathtt{Var}(\epsilon) = \mathtt{C}/\epsilon^2 \) and \( \mathtt{C} \) is a fixed constant. For the \textit{composition} of five mechanisms, four are pre-designed with noise perturbation using \( \mathcal{N}(0, \mathtt{Var}(\epsilon)) \). The defender's neural network \( G_5 \) selects \( \epsilon \) for the fifth mechanism, constrained by the BGBP response with \( \mathtt{BDP}[\vec{G}; 5\epsilon] \). 
We evaluate whether the target \( 5\epsilon \) can be approximately achieved if the attacker's performance aligns closely with that of a single \( 5\epsilon \)-DP mechanism (\textit{One Mechanism}), where the single \( 5\epsilon \)-DP mechanism is also perturbed by Gaussian noise \( \mathcal{N}(0, \mathtt{Var}(5\epsilon)) \). Our experimental results demonstrate that the BNGP strategy, constrained by BGBP response, successfully implements parameterized privacy in a generative manner.

\subsection{Additional Experiment: }

As shown in Figure \ref{fig:six_figuresap},  the privacy strength of the defense decreases (resp. increases) as $\kappa$ increases (resp. decreases), as we would expect, since $\kappa$ captures the tradeoff between privacy and utility.
Figure \ref{fig:fig4ap} demonstrates the performances of the Bayesian, fixed-threshold, and adaptive-threshold attackers under $\epsilon$-DP defense where $\epsilon=600$.
Large $\epsilon$ values are chosen because datasets with an extremely high number of variables require correspondingly larger epsilon parameters to avoid introducing prohibitive amounts of noise that would render the analysis unusable. \cite{venkatesaramani2023enabling}.
Similar to the scenarios under the Bayesian defense, the Bayesian attacker outperforms the LRT attackers under the $\epsilon$-DP.

\subsection{Network Configurations and Hyperparameters}

The \textbf{Defender} neural network is a generative model designed to process membership vectors and produce beacon modification decisions. The input layer feeds into two fully connected layers with batch normalization and activation functions applied after each layer. The first hidden layer uses ReLU activation, while the second hidden layer uses LeakyReLU activation. The output layer applies a scaled sigmoid activation function. The output of the Defender neural network is a real value between -0.5 and 0.5, which is guaranteed by the scaled sigmoid activation function.
All Defender neural networks were trained using the Adam optimizer with a learning rate of 0.001, weight decay of $0.00001$, and an ExponentialLR scheduler with a decay rate of $0.988$.

The \textbf{Attacker} neural network is a generative model designed to process beacons and noise to produce membership vectors. The input layer feeds into two fully connected layers with batch normalization and activation functions. The first hidden layer uses ReLU activation. The output layer applies a sigmoid activation function.
All Attacker models were trained using the Adam optimizer, a learning rate of $0.0001$, weight decay of $0.00001$, and an ExponentialLR scheduler with a decay rate of $0.988$.

The specific configurations for each model are provided in the tables below. Table \ref{Table:defender_scalar} shows the configurations of the neural network Defender under the Bayesian, the fixed-threshold, and the adaptive-threshold attackers when the trade-off parameter $\kappa$ is a vector (i.e., each $\kappa_j=\kappa$ for all $j\in Q$).
Table \ref{Table:defender_vector} shows the configurations of Defender when the trade-off parameter is a vector; i.e.,  $\vec{\kappa}=(\kappa_{j})_{j\in Q}$ where $\kappa_{j}=0$ for the $90\%$ of $5000$ SNVs and $\kappa_{j}=50$ for the remaining $10\%$.
Table \ref{Table:attacker_defender} lists the configurations of the neural network Attacker under the Bayesian, the fixed-threshold LRT, and the adaptive-threshold LRT defenders.
Table \ref{Table:attacker_DP}
lists the configurations of Attacker under the standard $\epsilon$-DP which induces the same $\vec{\kappa}$-weighted expected utility loss for the defender.

\begin{table}[H]
\centering
\caption{Bayesian Defender Configurations}\label{Table:defender}
\begin{minipage}{0.45\textwidth}
\centering
\subcaption{Defender with scalar $\kappa$}\label{Table:defender_scalar}
\begin{tabular}{|l|c|c|}
\hline
Layer            & Input Units & Output Units \\ \hline
Input Layer      & 830         & 1500         \\ \hline
Hidden Layer 1   & 1500        & 1100         \\ \hline
Hidden Layer 2   & 1100        & 500          \\ \hline
Output Layer     & 500         & 5000         \\ \hline
\end{tabular}
\end{minipage}
\hspace{0.05\textwidth}
\begin{minipage}{0.45\textwidth}
\centering
\subcaption{Defender with vector $\vec{\kappa}$}\label{Table:defender_vector}
\begin{tabular}{|l|c|c|}
\hline
Layer            & Input Units & Output Units \\ \hline
Input Layer      & 830         & 1000         \\ \hline
Hidden Layer 1   & 1000        & 3000         \\ \hline
Hidden Layer 2   & 3000        & 4600         \\ \hline
Output Layer     & 4600        & 5000         \\ \hline
\end{tabular}
\end{minipage}
\end{table}

\begin{table}[H]
\centering
\caption{Attacker Configurations}
\begin{subtable}[t]{0.45\textwidth}
\centering
\caption{Attacker vs. Defender}\label{Table:attacker_defender}
\begin{tabular}{|l|c|c|}
\hline
Layer            & Input Units & Output Units \\ \hline
Input Layer      & 5000        & 3400         \\ \hline
Hidden Layer 1   & 3400        & 2000         \\ \hline
Output Layer     & 2000        & 800          \\ \hline
\end{tabular}
\end{subtable}
\hspace{0.05\textwidth}
\begin{subtable}[t]{0.45\textwidth}
\centering
\caption{Bayesian Attacker vs. $\epsilon$-DP}\label{Table:attacker_DP}
\begin{tabular}{|l|c|c|}
\hline
Layer            & Input Units & Output Units \\ \hline
Input Layer      & 5000        & 3000         \\ \hline
Hidden Layer 1   & 3000        & 1000         \\ \hline
Output Layer     & 1000        & 800          \\ \hline
\end{tabular}
\end{subtable}
\end{table}

\subsection{AUC Values of ROC Curves with Standard Deviations}

Tables \ref{Table:AUC_different_attacks}, \ref{Table:AUC_different_defender}, and \ref{Table:implement_dp} show the AUC values of the ROC curves shown in the plots of the experiments.

\clearpage

\begin{table*}[!htbp]
\centering
\caption{AUC Values For Different Attackers Under Varying $\kappa$}
\label{Table:AUC_different_attacks}
\begin{tabular}{@{}lccc@{}}
\toprule
Attacker           & Figure 2a ($\kappa=0$) & Figure 1a and 2b ($\kappa=1.5$) & Figure 2c ($\kappa=50$) \\ \midrule
Bayesian attacker  & $0.5205 \pm 0.0055$    & $0.7253 \pm 0.0069$      & $0.8076 \pm 0.0040$     \\
Fixed-Threshold LRT attacker   & $0.5026 \pm 0.0062$    & $0.6214 \pm 0.0322$      & $0.7284 \pm 0.0089$     \\
Adaptive-Threshold LRT attacker   & $0.1552 \pm 0.0100$    & $0.1716 \pm 0.0144$      & $0.1719 \pm 0.0174$     \\ \bottomrule
\end{tabular}
\end{table*}

\begin{table*}[!htbp]
\centering
\caption{AUC Values of Attackers For Figures 1b to 1h}
\label{Table:AUC_different_defender}
\begin{tabular}{@{}cccc@{}}
\toprule
Figure & Scenarios             & AUC $\pm$ std        & Condition \\ \midrule
\multirow{3}{*}{1b} & Under Bayesian Defender            & $0.7237 \pm 0.0066$ & $\kappa = 1.5$ \\
                    & Under Fixed-threshold LRT Defender & $0.9124 \pm 0.0026$ & $\kappa = 1.5$ \\
                    & Under Adaptive-threshold LRT Defender   & $0.7487 \pm 0.0027$ & $\kappa = 1.5$ \\ 
\midrule
\multirow{2}{*}{1c} & Under Bayesian Defender            & $0.5318 \pm 0.0222$ &   $\vec{\kappa}$\\
                    & Under $\epsilon$-DP Defender       & $0.9153 \pm 0.0025$ &  $\vec{\kappa}$\\
\midrule
\multirow{4}{*}{1d} & Bayesian Attacker            & $0.5600 \pm 0.0040$ & $\kappa = 1.5$ \\
                    & Fix-LRT Attacker             & $0.5287 \pm 0.0052$ & $\kappa = 1.5$\\
                    & Adp-LRT Attacker             & $0.1431 \pm 0.0120$ & $\kappa = 1.5$ \\
                    & Score-Based Attacker & $0.1267 \pm 0.0207$ & $\kappa = 1.5$\\
\midrule
\multirow{4}{*}{1e} & Bayesian Attacker            & $0.6317 \pm 0.0050$ &  \\
                    & Fix-LRT Attacker             & $0.5865 \pm 0.0060$ & \\
                    & Adp-LRT Attacker             & $0.1722 \pm 0.0752$ &  \\
                    & Score-Based Attacker & $0.1223 \pm 0.0170$ & \\
\midrule
\multirow{4}{*}{1f} & Bayesian Attacker            & $0.5868 \pm 0.0035$ &  \\
                    & Fix-LRT Attacker             & $0.5615 \pm 0.0065$ & \\
                    & Adp-LRT Attacker             & $0.2076 \pm 0.0160$ &  \\
                    & Score-Based Attacker & $0.1229 \pm 0.0028$ & \\
\midrule
\multirow{4}{*}{1g} & Bayesian Attacker            & $1 \pm 0$ &  \\
                    & Fix-LRT Attacker             & $0.8618 \pm 0.0019$ & \\
                    & Adp-LRT Attacker             & $0.2221 \pm 0.0106$ &  \\
                    & Score-Based Attacker & $0.1542 \pm 0.0227$ & \\
\midrule
\multirow{3}{*}{1h} & Bayesian Attacker            & $0.7422 \pm 0.0085$ &  $\kappa = 1.5$\\
                    & Decision-Tree Attacker             & $0.6609 \pm 0.0110$ & $\kappa = 1.5$\\
                    & SVM Attacker             & $0.5226 \pm 0.0108$ &  $\kappa = 1.5$\\
\bottomrule
\end{tabular}
\end{table*}

\begin{table*}[!htbp]
\centering
\caption{AUC Values of Figure 1i}
\label{Table:implement_dp}
\begin{tabular}{@{}cc@{}}
\toprule
Scenarios             & AUC $\pm$ std         \\ \midrule
Composition ($\epsilon = 0.05$)            & $0.7387 \pm 0.0050$  \\
One Mechanism ($\epsilon = 0.05$) & $0.7427 \pm 0.0063$  \\
\midrule
Composition ($\epsilon = 0.1$)   & $0.8033 \pm  0.0057$ \\ 
One Mechanism ($\epsilon = 0.1$)  & $0.8241 \pm   0.0035$\\  
\midrule
Composition ($\epsilon = 0.3$)            & $0.8921\pm  0.0037 $  \\
One Mechanism ($\epsilon = 0.3$) & $0.9018 \pm 0.0033$  \\
\midrule
Composition ($\epsilon = 0.6$)   & $0.9108 \pm 0.0032$ \\ 
One Mechanism ($\epsilon = 0.6$)  & $0.9201 \pm 0.0032$\\ 
\midrule
Composition ($\epsilon = 1$)            & $0.9318 \pm 0.0031 $  \\
One Mechanism ($\epsilon = 1$) & $0.9373 \pm 0.0030$ \\
\bottomrule
\end{tabular}
\end{table*}

\section{Existing Frequentist Attack Models}\label{app:existence_lrt_attack}

\paragraph{Likelihood Ratio Test Attacks}

MIAs targeting genomic summary data releases are often framed as hypothesis testing problems \cite{sankararaman2009genomic,shringarpure2015privacy,raisaro2017addressing,venkatesaramani2021defending,venkatesaramani2023enabling},
where for each individual $k\in \mathcal{K}$, the attacker tests $H^{k}_{0}: b_{k}=1$ (i.e., the individual $k$ is in the dataset) versus $H^{k}_{1}: b_{k}=0$ (i.e., the individual $k$ is not).
Additionally, $\bar{p}_{j}$ denotes the frequency of the alternate allele at the $j$-th SNV in a reference population that is not included in the dataset $D$.

First, assume $\delta=0$. The attacker is assumed to have external knowledge of the genomic data for individuals $[K]$, in the form of $\bar{p}=(\bar{p}_{j})_{j\in Q}$ and $d=(d_{kj})_{k\in [K], j\in Q}$. 
The \textit{log-likelihood ratio statistic} (LRS) for each individual $k$ is given by \cite{sankararaman2009genomic}:
\begin{equation*}
    \begin{aligned}
        \mathtt{lrs}(d_{k}, x) = \sum\nolimits_{j\in Q}\left( d_{kj} \log\frac{\bar{p}_{j}}{x_{j}} + (1-d_{kj}) \log\frac{1-\bar{p}_{j}}{1-x_{j}}\right).
    \end{aligned}
\end{equation*}
An \textit{LRT attacker} performs MIA by testing $H^{k}_{0}$ against $H^{k}_{1}$ using $\mathtt{lrs}(d_{k}, x)$ for each $k\in[K]$.
The null hypothesis $H^{k}_{0}$ is rejected in favor of $H^{k}_{1}$ if $\mathtt{lrs}(d_{k}, x)\leq\tau$ for a threshold $\tau$, and $H^{k}_{0}$ is accepted if $\mathtt{lrs}(d_{k}, x)>\tau$.

Let $P^{k}_{0}(\cdot)\equiv\textup{Pr}(\cdot|H^{k}_{0})$ and $P^{k}_{1}(\cdot)\equiv\textup{Pr}(\cdot|H^{k}_{1})$ denote the probability distributions under $H_0$ and $H_1$, respectively.
\begin{definition}[$(\alpha_{\tau},\beta_{\tau})$-LRT Attack]
    An attacker performs \textup{$(\alpha_{\tau},\beta_{\tau})$-LRT Attack} if $P^{k}_{0}(\mathtt{lrs}(d_{k}, x)\leq\tau)= \alpha_{\tau}$ and $1-P^{k}_{1}(\mathtt{lrs}(d_{k}, x)\leq\tau)=\beta_{\tau}$, for all $k\in U$, where $\alpha_{\tau}$ is the \textup{significance level} and $1-\beta_{\tau}$ is the \textup{power} of the test with the \textup{threshold} $\tau$.
\end{definition}

Define the trade-off function \cite{dong2021gaussian}, 
\[
T[P^{k}_{0}, P^{k}_{1}](\alpha) \equiv \inf_{\tau}\{\beta_{\tau}: \alpha_{\tau}\leq \alpha\}.
\]
By Neyman-Pearson lemma \cite{neyman1933ix}, the LRT test is the uniformly most powerful (UMP) test for a given significance level. 
Specifically, for a given $\alpha_{\tau}$, there exists a threshold $\tau^{*}$ such that no other test with $\alpha_{\tau}\leq \alpha_{\tau^{*}}$ can achieve a strictly smaller $\beta_{\tau}<\beta_{\tau^{*}}$.
Hence, $T[P^{k}_{0}, P^{k}_{1}](\alpha_{\tau^{*}})=\beta_{\tau^{*}}$, for all $k\in \mathcal{K}$.
We refer to an $\alpha$-LRT as a UMP $(\alpha, \beta)$-LRT and will interchangeably add or omit the corresponding threshold notation as needed.

From the vNM defender's perspective, the expected privacy losses under an $\alpha$-LRT attack, without and with defense $g$, respectively, are given by
\[
\begin{aligned}
    &L^{o}(\tau^{o}, \alpha)\equiv\mathbb{E}\left[v(\tilde{s}, \tilde{b})\middle|\alpha\right]\\
    &= \sum\nolimits_{k}P^{k}_{1}\left[y_{k}(f(b,z),\tau^{o})=1\right] \theta(b_{k}=1)\\
    &=\sum\nolimits_{k}(1-\beta_{\tau^{o}})\theta(b_{k}=1),
\end{aligned}
\]
and
\[
\begin{aligned}
&L(g,\tau^{o},\alpha)\equiv\mathbb{E}\left[v(\tilde{s}, \tilde{b})\middle|g,\tau^{o},\alpha\right]\\
&= \sum\nolimits_{k}P^{k}_{1}\left[y_{k}(r,\tau^{o})=1\middle|g\right] \theta(b_{k}=1),
\end{aligned}
\]
where $y_{k}(x,\tau^{o})\equiv \mathbf{1}\left\{\mathtt{lrs}(d_{k}, x)\geq \tau^{o}\right\}$ is the indicator function for the likelihood ratio statistic, and 
\[
P^{k}_{1}[y_{k}(r,\tau^{o})=1|g]\equiv \int_{r} \mathbf{1}\left\{y_{k}(r,\tau^{o})=1\right\} \rho_{D}(r|b)dr.
\]
Here, $\tau^{o}$ is the threshold associated with the $\alpha$-LRT.
where $y_{k}(x,\tau^{o})\equiv \mathbf{1}\left\{\mathtt{lrs}(d_{k}, x)\geq \tau^{o}\right\}$ is the indicator function for the likelihood ratio statistic, and 
\[
P^{k}_{1}[y_{k}(r,\tau^{o})=1|g]\equiv \int_{r} \mathbf{1}\left\{y_{k}(r,\tau^{o})=1\right\} \rho_{D}(r|b)dr.
\]
Here, $\tau^{o}$ is the threshold associated with the $\alpha$-LRT.

\paragraph{Fixed-Threshold LRT Attack \cite{sankararaman2009genomic, shringarpure2015privacy, venkatesaramani2021defending, venkatesaramani2023enabling}}
A \textit{fixed(-threhsold)} LRT attacker performs MIA without accounting for any privacy defense strategies. Such an attacker selects a fixed threshold $\tau^{o}$ that balances Type-I and Type-II errors, resulting in a UMP $\alpha$-LRT test in the absence of defense. This approximation can be achieved by simulating Beacons on publicly available datasets or synthesized data using alternate allele frequencies (AAFs) \cite{venkatesaramani2023enabling}.

Given a fixed threshold $\tau^{o}$, let
\[
L^{\alpha}_{\textup{Fixed}}(g) \equiv L(g, \tau^{o},\alpha_{\tau^{o}})
\]

The defender’s optimal strategy against the naive $\alpha_{\tau^{o}}$-LRT attack is given by solving:
\begin{equation}\tag{\texttt{FixedLRT}}\label{eq:nariveLRT}
    \min\nolimits_{g} L^{\alpha}_{\textup{Fixed}}(g)  + \kappa \mathbb{E}\left[\ell_{U}(\|\delta\|_{\mathtt{p}})\|\middle|g, \tau^{o}, \alpha_{\tau^{o}}\right],
\end{equation}
where $\mathbb{E}\left[\ell_{U}(\|\delta\|_{\mathtt{p}})\|\middle|g, \tau^{o}, \alpha_{\tau^{o}}\right]$ is induced expected utility loss.

Let $\beta^{k}(\tau, g, \alpha)\equiv 1-P^{k}_{1}[y_{k}(r,\tau)=1|g]$ denote the actual Type-II error under the defense strategy $g$ for the naive $\alpha$-LRT attack.
The defender can reduce privacy loss by choosing $g$ to increase $\beta^{k}(\tau, g, \alpha)$ for all $k\in U$.
For the defense strategy $g^{\dagger}_{D}$ that solves (\ref{eq:nariveLRT}) to be effective in reducing privacy loss, it must be implemented in a \textit{stealthy} manner.

\paragraph{Adaptive-Threshold LRT Attack \cite{venkatesaramani2021defending,venkatesaramani2023enabling}}  
In an \textit{adaptive-threshold LRT} attack, the attacker is aware of the defense strategy and attempts to distinguish individuals in \( \mathcal{K} \) from those in the reference population \( \bar{\mathcal{K}} \) (individuals not in \( \mathcal{K} \)). Let \( \bar{\mathcal{K}}^{(N)} \subset \bar{\mathcal{K}} \) represent the set of \( N \) individuals in \( \bar{\mathcal{K}} \) with the lowest LRS values. The \textit{adaptive threshold} is defined as \( \tau^{(N)}(x) = \frac{1}{N} \sum_{i \in \bar{\mathcal{K}}^{(N)}} \mathtt{lrs}(d_{i}, x) \). The null hypothesis \( H_{0} \) is rejected if \( \mathtt{lrs}(d_{k}, x) \leq \tau^{(N)}(x) \).

The defender's problem is then:
\[
    \min_{g} L^{\alpha_{\tau^{(N)}(x)}}_{\textup{Adp}}(g) + \mathbb{E}\left[\ell_{U}(\|\delta\|_{\mathtt{p}})\|\middle|g, \tau^{(N)}(x), \alpha_{\tau^{(N)}(x)}\right],
\tag{\texttt{AdaptLRT}} \label{eq:adaptiveLRT}
\]
where \( \alpha_{\tau^{(N)}(x)} \) is the Type-I error associated with the adaptive threshold \( \tau^{(N)}(x) \), and
\[
L^{\alpha_{\tau^{(N)}(x)}}_{\textup{Adp}}(g) \equiv L(g, \tau^{(N)}(x), \alpha_{\tau^{(N)}(x)}).
\]
The WCPL $L^{\alpha_{\tau^{(N)}(x)}}_{\textup{Adp}}(g)$ in Section \ref{sec:case_study} has \( \mathbb{E}[\alpha_{\tau^{(N)}(\tilde{r})}] = \alpha \).

\paragraph{Optimal LRT Attack}  
Let \( P^{k}_{0}(g) = P^{k}_{0}[\cdot | g] \) and \( P^{k}_{1}(g) = P^{k}_{1}[\cdot | g] \) denote the probability distributions under \( H^{k}_{0} \) and \( H^{k}_{1} \), respectively, in the presence of defense \( g \). The \textit{worst-case privacy loss} (WCPL) for the defender occurs when the attacker’s hypothesis test achieves \( \beta^{k}(\tau^{*}, g, \alpha) = T[P^{k}_{0}(g), P^{k}_{1}(g)](\alpha) \) for some threshold \( \tau^{*} \), corresponding to a UMP test under \( g \). We refer to these as \textit{optimal \( \alpha \)-LRT attacks}.

The defender’s optimal strategy against such attacks solves the following problem:
\begin{equation}\tag{\texttt{OptLRT}} \label{eq:optimalLRT}
    \begin{aligned}
    &\min_{g} L^{\alpha}_{\textup{Opt-LRT}}(g) + \kappa \mathbb{E}\left[\ell_{U}(\|\delta\|_{\mathtt{p}})\|\middle|g_D, \tau^{*}, \alpha\right], \quad \\
    &\text{s.t.} \quad \beta^{k}(\tau^{*}, g, \alpha) = T[P^{k}_{0}(g), P^{k}_{1}(g)](\alpha),
\end{aligned}
\end{equation}
where 
\[
L^{\alpha}_{\textup{Opt-LRT}}(g) \equiv L(g, \tau^{*}, \alpha).
\]
By the Neyman-Pearson lemma, the \( \alpha \)-LRT with likelihood ratio statistics 
\[
\mathtt{lrs}(d_{k}, r; g) \equiv \sum_{j \in Q} \frac{\rho_{D}(r | b_{k} = 0, b_{-k})}{\rho_{D}(r | b_{k} = 1, b_{-k})},
\]
for all \( k \in \mathcal{K} \) is optimal, attaining \( \beta^{k}(\tau^{*}, g, \alpha) = T[P^{k}_{0}(g), P^{k}_{1}(g)](\alpha) \).

Furthermore, the defense \( g \) obtained by solving (\ref{eq:optimalLRT}) is robust against adaptive-threshold LRT attacks.

\vfill

\end{document}